\documentclass[11pt]{article}
\usepackage{CJK}
\usepackage{latexsym,bm}
\usepackage{xypic}
\usepackage{amsmath}
\usepackage{wasysym}
\usepackage{indentfirst}
\usepackage{amssymb}
\usepackage{dsfont}
\usepackage{amsthm}
\begin{document}
\newcommand{\fr}[2]{\frac{\;#1\;}{\;#2\;}}
\newtheorem{theorem}{Theorem}[section]
\newtheorem{lemma}{Lemma}[section]
\newtheorem{proposition}{Proposition}[section]
\newtheorem{corollary}{Corollary}[section]
\newtheorem{conjecture}{Conjecture}[section]
\newtheorem{remark}{Remark}[section]
\newtheorem{definition}{Definition}[section]
\newtheorem{example}{Example}[section]
\newtheorem{notation}{Notation}[section]
\numberwithin{equation}{section}
\newcommand{\Aut}{\mathrm{Aut}\,}
\newcommand{\CSupp}{\mathrm{CSupp}\,}
\newcommand{\Supp}{\mathrm{Supp}\,}
\newcommand{\rank}{\mathrm{rank}\,}
\newcommand{\col}{\mathrm{col}\,}
\newcommand{\len}{\mathrm{len}\,}
\newcommand{\ind}{\mathrm{ind}\,}
\newcommand{\leftlen}{\mathrm{leftlen}\,}
\newcommand{\rightlen}{\mathrm{rightlen}\,}
\newcommand{\length}{\mathrm{length}\,}
\newcommand{\wt}{\mathrm{wt}\,}
\newcommand{\diff}{\mathrm{diff}\,}
\newcommand{\lcm}{\mathrm{lcm}\,}
\newcommand{\dom}{\mathrm{dom}\,}
\newcommand{\SUPP}{\mathrm{SUPP}\,}
\newcommand{\supp}{\mathrm{supp}\,}
\newcommand{\End}{\mathrm{End}\,}
\newcommand{\Hom}{\mathrm{Hom}\,}
\newcommand{\ran}{\mathrm{ran}\,}
\newcommand{\Mat}{\mathrm{Mat}\,}
\newcommand{\rk}{\mathrm{rk}\,}
\newcommand{\rs}{\mathrm{rs}\,}
\newcommand{\piv}{\mathrm{piv}\,}
\newcommand{\perm}{\mathrm{perm}\,}
\newcommand{\inv}{\mathrm{inv}\,}
\newcommand{\orb}{\mathrm{orb}\,}
\newcommand{\id}{\mathrm{id}\,}
\newcommand{\soc}{\mathrm{soc}\,}
\newcommand{\Jac}{\mathrm{Jac}\,}
\newcommand{\GL}{\mathrm{GL}\,}

\title{Isometries and MacWilliams Extension Property for Weighted Poset Metric$^\ast$}

\author{Yang Xu$^1$ \,\,\,\,\,\, Haibin Kan$^2$\,\,\,\,\,\,Guangyue Han$^3$}

\maketitle

\renewcommand{\thefootnote}{\fnsymbol{footnote}}

\footnotetext{\hspace*{-12mm} \begin{tabular}{@{}r@{}p{13.4cm}@{}}
$^\ast$ & A preliminary version of this paper has been presented in IEEE International Symposium on Information Theory (ISIT) 2022.\\
$^1$ & Shanghai Key Laboratory of Intelligent Information Processing, School of Computer Science, Fudan University,
Shanghai 200433, China.\\
&Department of Mathematics, Faculty of Science, The University of Hong Kong, Pokfulam Road, Hong Kong, China. {E-mail:12110180008@fudan.edu.cn} \\
$^2$ & Shanghai Key Laboratory of Intelligent Information Processing, School of Computer Science, Fudan University,
Shanghai 200433, China.\\
&Shanghai Engineering Research Center of Blockchain, Shanghai 200433, China.\\
&Yiwu Research Institute of Fudan University, Yiwu City, Zhejiang 322000, China. {E-mail:hbkan@fudan.edu.cn} \\
$^3$ & Department of Mathematics, Faculty of Science, The University of Hong Kong, Pokfulam Road, Hong Kong, China. {E-mail:ghan@hku.hk} \\
\end{tabular}}

\vskip 3mm

{\hspace*{-6mm}\textbf{Abstract}---Let $\mathbf{H}$ be the cartesian product of a family of left modules over a ring $S$, indexed by a finite set $\Omega$. We are concerned with the $(\mathbf{P},\omega)$-weight on $\mathbf{H}$, where $\mathbf{P}=(\Omega,\preccurlyeq_{\mathbf{P}})$ is a poset and $\omega:\Omega\longrightarrow\mathbb{R}^{+}$ is a weight function. We characterize the group of $(\mathbf{P},\omega)$-weight isometries of $\mathbf{H}$, and give a canonical decomposition for semi-simple subcodes of $\mathbf{H}$ when $\mathbf{P}$ is hierarchical. We then study the MacWilliams extension property (MEP) for $(\mathbf{P},\omega)$-weight. We show that the MEP implies the unique decomposition property (UDP) of $(\mathbf{P},\omega)$, which further implies that $\mathbf{P}$ is hierarchical if $\omega$ is identically $1$. For the case that either $\mathbf{P}$ is hierarchical or $\omega$ is identically $1$, we show that the MEP for $(\mathbf{P},\omega)$-weight can be characterized in terms of the MEP for Hamming weight, and give necessary and sufficient conditions for $\mathbf{H}$ to satisfy the MEP for $(\mathbf{P},\omega)$-weight when $S$ is an Artinian simple ring (either finite or infinite). When $S$ is a finite field, in the context of $(\mathbf{P},\omega)$-weight, we compare the MEP with other coding theoretic properties including the MacWilliams identity, Fourier-reflexivity of partitions and the UDP, and show that the MEP is strictly stronger than all the rest among them.

\section{Introduction}

The notion of \textit{weighted poset metric} has been introduced by Hyun, Kim and Park in \cite{28} for binary field alphabet. A weighted poset metric is determined by a poset and a weight function, both defined on the coordinate set. In \cite{28}, the authors have classified all the weighted posets and directed graphs that admit the extended Hamming code $\widetilde{\mathcal{H}}_{3}$ to be a $2$-perfect code, and relevant results for more general $\widetilde{\mathcal{H}}_{k}$, $k\geqslant3$ have also been established. It has also been shown in \cite{28} that weighted poset metric can be viewed as an algebraic version of the directed graph metric introduced by Etzion, Firer and Machado in \cite{18}, and we refer the reader to [28, Sections I, II] and [18, Section III] for the connections between these two metrics.

Weighted poset metric is rather general in the sense that it includes some well studied metrics as special cases, such as poset metric (see \cite{7,27,38}) and weighted Hamming metric (see \cite{5}). Since the weight function takes values on each coordinate position, weighted poset metric can be useful to model some specific kind of channels for which the error probability depends on a codeword position, i.e., the distribution of errors is nonuniform, and can also be useful to perform bitwise or messagewise unequal error protection (see, e.g., the abstract of \cite{5} and [18, Section 1, Paragraph 6]).

More recently in \cite{35}, Machado and Firer have proposed and studied labeled-poset-block metric for finite field alphabet, which is a generalization of both the weighted poset metric in \cite{28} and the directed graph metric in \cite{18}. In \cite{35}, the authors have studied the group of linear isometries, the MacWilliams identity and the MacWilliams extension property (MEP) for labeled-poset-block metric. In particular, for binary field alphabet, they have given a necessary and sufficient condition for the MEP when the poset is hierarchical.

In this paper, we consider weighted poset metric for module alphabet. More precisely, the ambient space $\mathbf{H}=\prod_{i\in\Omega}H_{i}$ is the cartesian product of a family of left modules over a ring $S$, indexed by a finite set $\Omega$. This further generalizes the labeled-poset-block metric in \cite{35}. We will study the group of isometries and the MacWilliams extension property (MEP) for weighted poset metric.

Groups of linear isometries for various metrics have been studied extensively in the literature, and have been characterized for Rosenbloom-Tsfasman weight by Lee in \cite{31}, for crown weight by Cho and Kim in \cite{9}, for poset metric by Panek, Firer, Kim and Hyun in \cite{39}, for poset-block metric by Alves, Panek and Firer in \cite{1}, for directed graph metric by Etzion, Firer and Machado in \cite{18}, for combinatorial metric by Pinheiro, Machado and Firer in \cite{42}, and for labeled-poset-block metric by Machado and Firer in \cite{35}. We also refer the reader to \cite{35,40} for isometries for two general metrics in poset space.

In 1962, MacWilliams proved in \cite{37} that for a finite field $\mathbb{F}$ and a positive integer $n$, any Hamming weight preserving map between two linear codes extends to a Hamming weight isometry of $\mathbb{F}^{n}$ (also see \cite{6,46} for other proofs). Such a property, henceforth referred to as the MacWilliams extension property (MEP), has since been extended, generalized and discussed extensively in the literature: with respect to other weights and metrics such as symmetrized weight composition, homogeneous weight, bi-invariant weight over finite rings, rank metric, poset metric, combinatorial metric, directed graph metric and labeled-poset-block; with respect to codes over ring and module alphabets (both finite and infinite); and with respect to local-global property for subgroups of the general linear group, along with partitions of finite modules; see, among many others, \cite{3}, \cite{6}, \cite{12}--\cite{18}, \cite{21}--\cite{25}, \cite{29}, \cite{33}--\cite{35}, \cite{42}, \cite{44}, \cite{46}--\cite{49}.

The remainder of the paper is organized as follows. In Section 2, we give some definitions, notations and basic facts of weighted poset metric, the MEP and modules. In Section 3, we study the group of isometries for weighted poset metric. We consider a slightly more general case, and derive relevant results for weighted poset metric as a consequence. In Section 4, we study the MEP for $\mathbf{P}$-support, where $\mathbf{P}$ is a poset on $\Omega$. This is a special case of the MEP for general weighted poset metric. When $\mathbf{P}$ is hierarchical, we give a necessary and sufficient condition for $\mathbf{H}$ to satisfy the MEP, and give a canonical decomposition for semi-simple codes. Some other sufficient conditions for the MEP are also given for possibly non-hierarchical $\mathbf{P}$. In Section 5, we study \textit{isometry equation}, a notion that has been introduced by Dyshko to study the MEP for various weights (see \cite{13}--\cite{17} and [29, Lemma 4.4]). We derive the minimal length of nontrivial solutions to the isometry equation with respect to a finite lattice, which is then used to derive some sufficient conditions for Hamming weight preserving maps to be extendable.

In Section 6, we consider the MEP for $(\mathbf{P},\omega)$-weight for a poset $\mathbf{P}$ and a weight function $\omega:\Omega\longrightarrow\mathbb{R}^{+}$. In Section 6.1, we first show that with some seemingly relatively mild assumptions, the MEP implies the unique decomposition property (UDP) for $(\mathbf{P},\omega)$, which further implies that $\mathbf{P}$ is hierarchical if $\omega$ is identically $1$. Next, we focus on the case that either $\mathbf{P}$ is hierarchical or $\omega$ is identically $1$. We show that for such cases, the MEP for $(\mathbf{P},\omega)$-weight can be characterized in terms of the MEP for Hamming weight. We then give some explicit sufficient conditions for the MEP, and derive a necessary and sufficient condition for the MEP when $S$ is an Artinian simple ring (either finite or infinite). In Section 6.2, $\mathbf{H}$ is supposed to be a finite vector space. We compare the MEP with some other coding-theoretic properties including MacWilliams identity, Fourier-reflexivity of partitions, the UDP and that whether the group of isometries acts transitively on codewords with the same weight (see \cite{3,10,18,20,21,22,30,33,34,35,41,42,50,52}), and show that the MEP is strictly stronger than all the others.

\section{Preliminaries}
\setlength{\parindent}{2em}
We begin with some notations that are used throughout the remainder of the paper. For any $a,b\in\mathbb{Z}$, we let $[a,b]$ denote the set of all the integers between $a$ and $b$, i.e., $[a,b]=\{i\in\mathbb{Z}\mid a\leqslant i\leqslant b\}$. We also let $S$ be an associative ring with the multiplicative identity $1_{S}$, $\Omega$ be a nonempty finite set, $(H_{i}\mid i\in \Omega)$ be a family of left $S$-modules, and let
\begin{equation}\mathbf{H}=\prod_{i\in\Omega}H_{i}.\end{equation}
Any $S$-submodule of $\mathbf{H}$ is referred to as a \textit{linear code}. For any \textit{codeword} $\beta\in\mathbf{H}$, we let $\supp(\beta)$ denote the set
\begin{equation}\supp(\beta)\triangleq\{i\in\Omega\mid \beta_{i}\neq0\}.\end{equation}
For $i\in \Omega$, define $\pi_{i}:\mathbf{H}\longrightarrow H_{i}$ as $\pi_{i}(\alpha)=\alpha_{i}$, and define $\eta_{i}:H_{i}\longrightarrow \mathbf{H}$ as
\begin{equation}\forall~a\in H_{i}:\supp(\eta_{i}(a))\subseteq\{i\},~(\eta_{i}(a))_{i}=a.\end{equation}
For any $I\subseteq\Omega$, define $\delta(I)\subseteq\mathbf{H}$ as
\begin{equation}\delta(I)=\{\beta\in \mathbf{H}\mid\supp(\beta)\subseteq I\}.\end{equation}
It is known that $\End_{S}(\mathbf{H})$ and $\prod_{(i,j)\in\Omega\times\Omega}\Hom_{S}(H_{i},H_{j})$ can be identified via the one-to-one correspondence $\varphi\mapsto(\pi_{j}\circ\varphi\circ\eta_{i}\mid (i,j)\in\Omega\times\Omega)$ (see [2, Chapter 2, Section 6]).

\subsection{Weighted poset metric}

Throughout this subsection, we let $\mathbf{P}=(\Omega,\preccurlyeq_{\mathbf{P}})$ be a poset. A subset $B\subseteq\Omega$ is said to be an \textit{ideal} of $\mathbf{P}$ if for any $b\in B$ and $a\in\Omega$, $a\preccurlyeq_{\mathbf{P}}b$ implies that $a\in B$. We let $\mathcal{I}(\mathbf{P})$ denote the set of all the ideals of $\mathbf{P}$. For $B\subseteq\Omega$, we let $\langle B\rangle_{\mathbf{P}}$ denote the ideal $\{a\in\Omega\mid \exists~b\in B~s.t.~a\preccurlyeq_{\mathbf{P}}b\}$. In addition, $B$ is said to be a \textit{chain} in $\mathbf{P}$ if for any $a,b\in B$, either $a\preccurlyeq_{\mathbf{P}}b$ or $b\preccurlyeq_{\mathbf{P}}a$ holds, and $B$ is said to be an \textit{anti-chain} in $\mathbf{P}$ if for any $a,b\in B$, $a\preccurlyeq_{\mathbf{P}}b$ implies that $a=b$. For any $y\in\Omega$, we let $\len_{\mathbf{P}}(y)$ denote the largest cardinality of a chain in $\mathbf{P}$ containing $y$ as its greatest element. The \textit{dual poset} of $\mathbf{P}$ will be denoted by $\mathbf{\overline{P}}$, where $u\preccurlyeq_{\mathbf{\overline{P}}}v\Longleftrightarrow v\preccurlyeq_{\mathbf{P}}u$ for all $u,v\in\Omega$. The set of all the order automorphisms of $\mathbf{P}$ will be denoted by $\Aut(\mathbf{P})$.

\setlength{\parindent}{0em}
\begin{definition}
{{\bf{(1)}}\,\,$\mathbf{P}$ is said to be hierarchical if for any $u,v\in\Omega$ with $\len_{\mathbf{P}}(u)+1\leqslant\len_{\mathbf{P}}(v)$, it holds that $u\preccurlyeq_{\mathbf{P}}v$.

{\bf{(2)}}\,\,For $\omega:\Omega\longrightarrow\mathbb{R}^{+}$, we say that $(\mathbf{P},\omega)$ satisfies the unique decomposition property (UDP) if for any $I,J\in\mathcal{I}(\mathbf{P})$ with $\sum_{i\in I}\omega(i)=\sum_{j\in J}\omega(j)$, there exists $\lambda\in\Aut(\mathbf{P})$ such that $J=\lambda[I]$ and $\omega(\lambda(i))=\omega(i)$ for all $i\in\Omega$.
}
\end{definition}

\setlength{\parindent}{2em}
We note that hierarchical poset has been extensively studied for poset codes (see \cite{3,10,18,19,20,30,33,34,35,41,50}), and the UDP has been proposed in [18, Definition 2] and [35, Definition 11] in slightly different forms.

\setlength{\parindent}{2em}
Now we fix $\omega:\Omega\longrightarrow\mathbb{R}^{+}$. Following \cite{28}, $(\mathbf{P},\omega)$ is referred to as an \textit{$\omega$-weighted poset}. For any $\beta\in\mathbf{H}$, the $(\mathbf{P},\omega)$-weight of $\beta$ is defined as
\begin{equation}\wt_{(\mathbf{P},\omega)}(\beta)\triangleq\sum_{i\in\langle\supp(\beta)\rangle_{\mathbf{P}}}\omega(i).\end{equation}
It has been proven in \cite{28} that $d_{(\mathbf{P},\omega)}:\mathbf{H}\times \mathbf{H}\longrightarrow \mathbb{R}$ defined as \begin{equation}\mbox{$d_{(\mathbf{P},\omega)}(\alpha,\beta)=\wt_{(\mathbf{P},\omega)}(\beta-\alpha)$}\end{equation}
induces a metric on $\mathbf{H}$, which will henceforth be referred to as a \textit{weighted poset metric}. We note that if $\omega$ is identically $1$, then (2.5) recovers the definition of $\mathbf{P}$-weight (see \cite{7,27,38}), i.e.,
\begin{equation}\forall~\beta\in\mathbf{H}:\wt_{\mathbf{P}}(\beta)\triangleq|\langle\supp(\beta)\rangle_{\mathbf{P}}|.\end{equation}
If $\mathbf{P}$ is an anti-chain, then (2.6) recovers the notion of weighted Hamming metric (see \cite{5}). In addition, if $S$ is a finite field, $\mathbf{H}$ is finite and $\omega$ is integer-valued, then (2.6) becomes the labeled-poset-block metric proposed in \cite{35}.

\setlength{\parindent}{0em}
\begin{definition}
{\bf{(1)}}\,\,For a linear code $C\subseteq\mathbf{H}$ and $f\in\Hom_{S}(C,\mathbf{H})$, we say that $f$ preserves $(\mathbf{P},\omega)$-weight if $\wt_{(\mathbf{P},\omega)}(f(\alpha))=\wt_{(\mathbf{P},\omega)}(\alpha)$ for all $\alpha\in C$. Any $S$-module automorphism of $\mathbf{H}$ that preserves $(\mathbf{P},\omega)$-weight is referred to as a \textit{$(\mathbf{P},\omega)$-weight isometry} of $\mathbf{H}$. We let $\GL_{(\mathbf{P},\omega)}(\mathbf{H})$ denote the set of all the $(\mathbf{P},\omega)$-weight isometries of $\mathbf{H}$. Moreover, we say that $\mathbf{H}$ satisfies the MacWilliams extension property (MEP) for $(\mathbf{P},\omega)$-weight if for any linear code $C\subseteq\mathbf{H}$ and $f\in\Hom_{S}(C,\mathbf{H})$ such that $f$ preserves $(\mathbf{P},\omega)$-weight, there exists $\varphi\in\GL_{(\mathbf{P},\omega)}(\mathbf{H})$ with $\varphi\mid_{C}=f$.

{\bf{(2)}}\,\,For a linear code $C\subseteq\mathbf{H}$ and $f\in\Hom_{S}(C,\mathbf{H})$, we say that $f$ preserves $\mathbf{P}$-support if $\langle\supp(f(\alpha))\rangle_{\mathbf{P}}=\langle\supp(\alpha)\rangle_{\mathbf{P}}$ for all $\alpha\in C$. We let $\GL_{\mathbf{P}}(\mathbf{H})$ denote the set of all the $S$-module automorphisms of $\mathbf{H}$ that preserve $\mathbf{P}$-support. Moreover, we say that $\mathbf{H}$ satisfies the MEP for $\mathbf{P}$-support if for any linear code $C\subseteq\mathbf{H}$ and $f\in\Hom_{S}(C,\mathbf{H})$ such that $f$ preserves $\mathbf{P}$-support, there exists $\varphi\in\GL_{\mathbf{P}}(\mathbf{H})$ with $\varphi\mid_{C}=f$.
\end{definition}

\begin{remark}
The MEP for $\mathbf{P}$-support is indeed a special case of the MEP for $(\mathbf{P},\omega)$-weight. More precisely, let $\sigma:\Omega\longrightarrow[0,|\Omega|-1]$ be a bijection, and set $\omega:\Omega\longrightarrow\mathbb{R}^{+}$ as $\omega(i)=2^{\sigma(i)}$. By (2.5), we infer that for any $\alpha,\beta\in\mathbf{H}$, $\wt_{(\mathbf{P},\omega)}(\alpha)=\wt_{(\mathbf{P},\omega)}(\beta)\Longleftrightarrow\langle\supp(\alpha)\rangle_{\mathbf{P}}=\langle\supp(\beta)\rangle_{\mathbf{P}}$. Hence a map preserves $(\mathbf{P},\omega)$-weight if and only if it preserves $\mathbf{P}$-support, and consequently, $\mathbf{H}$ satisfies the MEP for $(\mathbf{P},\omega)$-weight if and only if $\mathbf{H}$ satisfies the MEP for $\mathbf{P}$-support.
\end{remark}

\subsection{Some remarks for modules}
\setlength{\parindent}{2em}
In this subsection, we collect some definitions and notations for modules, most of which are known and can be found in \cite{2,49}.

First of all, recall that the ring $S$ is said to be \textit{Artinian simple} if $S$ is both left Artinian and simple. For $e,k\in\mathbb{Z}^{+}$, we let $Mat_{e,k}(S)$ denote the set of all the matrices over $S$ with $e$ rows and $k$ columns, and write $Mat_{e}(S)\triangleq Mat_{e,e}(S)$. By the celebrated Wedderburn-Artin Theorem, $S$ is Artinian simple if and only if $S$ is isomorphic to $Mat_{e}(\mathbb{D})$ for some division ring $\mathbb{D}$ and $e\in\mathbb{Z}^{+}$ (see, e.g, [2, Theorems 13.6 and 13.7]).

Next, consider two left $S$-modules $X$ and $Y$. We write $X\cong Y$ if $X$ and $Y$ are isomorphic as left $S$-modules. $Y$ is said to be \textit{$X$-injective} if for any $S$-submodule $A\subseteq X$ and $f\in\Hom_{S}(A,Y)$, there exists $g\in\Hom_{S}(X,Y)$ with $g\mid_{A}=f$. We will say that $Y$ is \textit{strong pseudo-injective} if for any $S$-submodule $B\subseteq Y$ and any injective $h\in\Hom_{S}(B,Y)$, there exists $\tau\in\Aut_{S}(Y)$ with $\tau\mid_{B}=h$. We also let $\soc_{S}(Y)$ denote the socle of $Y$, i.e., the largest semi-simple $S$-submodule of $Y$.

\setlength{\parindent}{2em}
Now we fix a poset $\mathbf{P}=(\Omega,\preccurlyeq_{\mathbf{P}})$ and $\omega:\Omega\longrightarrow\mathbb{R}^{+}$. For convenience, we collect five conditions which will appear frequently in our discussion.

\setlength{\parindent}{0em}
\begin{definition}
{{\bf{(1)}}\,\,We say that $\mathbf{H}$ satisfies Condition (A) if $H_{i}$ is strong pseudo-injective for all $i\in\Omega$.

{\bf{(2)}}\,\,We say that $(\mathbf{H},\mathbf{P})$ satisfies Condition (B) if for any $k,l\in\Omega$ such that $k\preccurlyeq_{\mathbf{P}}l$, $k\neq l$, it holds true that $H_{k}$ is $H_{l}$-injective.

{\bf{(3)}}\,\,We say that $\mathbf{H}$ satisfies Condition (C) if there exists $\xi\in\mathbf{H}$ such that $\xi\neq0$ and for any $k,l\in\Omega$, it holds that
$(\forall~a\in S:a\cdot\xi_{k}=0\Longleftrightarrow a\cdot\xi_{l}=0)$.

{\bf{(4)}}\,\,We say that $(\mathbf{H},(\mathbf{P},\omega))$ satisfies Condition (D) if $(\mathbf{P},\omega)$ satisfies the UDP, and for any $u,v\in\Omega$ such that $\len_{\mathbf{P}}(u)=\len_{\mathbf{P}}(v)$, $\omega(u)=\omega(v)$, it holds that $H_{u}\cong H_{v}$.

{\bf{(5)}}\,\,We say that $(\mathbf{H},\mathbf{P})$ satisfies Condition (E) if $\mathbf{P}$ is hierarchical, and for any $u,v\in\Omega$ such that $\len_{\mathbf{P}}(u)=\len_{\mathbf{P}}(v)$, it holds that $H_{u}\cong H_{v}$.
}
\end{definition}

\begin{remark}
One can check that $\mathbf{H}$ satisfies Condition (C) if and only if there exists a left $S$-module $B$ such that $B\neq\{0\}$ and for any $i\in\Omega$, $B$ is isomorphic to some $S$-submodule of $H_{i}$. Condition (C) seems to be relatively mild. In particular, if $S$ is an Artinian simple ring and $H_{i}\neq\{0\}$ for all $i\in\Omega$, then $\mathbf{H}$ satisfies Condition (C). We will show in Section 6 that if $\mathbf{H}$ satisfies Condition (C), then Conditions (A), (B), (D) (or (E)) are all necessary conditions for the MEP.
\end{remark}

\setlength{\parindent}{2em}
We end this subsection by noting that Conditions D and E are closely related. The following lemma is a consequence of [34, Theorem 3], and will be used frequently in our discussion.

\setlength{\parindent}{0em}
\begin{lemma}
{Suppose that $\omega$ is identically $1$. Then, $(\mathbf{P},\omega)$ satisfies the UDP if and only if $\mathbf{P}$ is hierarchical. Consequently, $(\mathbf{H},(\mathbf{P},\omega))$ satisfies Condition (D) if and only if $(\mathbf{H},\mathbf{P})$ satisfies Condition (E).
}
\end{lemma}

\section{Group of isometries for $(\mathbf{P},\omega)$-weight}

\setlength{\parindent}{2em}
Throughout this section, we let $K\triangleq\{i\in\Omega\mid H_{i}\neq\{0\}\}$, and fix a poset $\mathbf{P}=(\Omega,\preccurlyeq_{\mathbf{P}})$.

We first consider a slightly more general case. More precisely, we fix $(Y,\curlyeqprec)$ such that $\curlyeqprec$ is an anti-symmetric relation on $Y$, and fix $\varpi:2^{\Omega}\longrightarrow Y$ satisfying the following three conditions:
\begin{equation}\forall~B\subseteq\Omega:\varpi(B)=\varpi(\langle B\rangle_{\mathbf{P}}).\end{equation}
\begin{equation}\forall~I,J\in\mathcal{I}(\mathbf{P}):I\subseteq J\Longrightarrow\varpi(I)\curlyeqprec\varpi(J).\end{equation}
\begin{equation}\forall~I\in\mathcal{I}(\mathbf{P}),\forall~u\in I:\varpi(I)=\varpi(\{u\})\Longrightarrow I=\langle \{u\}\rangle_{\mathbf{P}}.\end{equation}
Now we define $T\leqslant\Aut(\mathbf{P})$ and $G\leqslant\Aut_{S}(\mathbf{H})$ as follows:
$$T=\{\mu\in\Aut(\mathbf{P})\mid\mu\mid_{\Omega-K}=\id_{\Omega-K};\forall~I\subseteq K,\varpi(\mu[I])=\varpi(I); \forall i\in\Omega,H_{i}\cong H_{\mu(i)}\},$$
$$G=\{\varphi\in\Aut_{S}(\mathbf{H})\mid\forall~\alpha\in\mathbf{H},\varpi(\supp(\varphi(\alpha)))=\varpi(\supp(\alpha))\}.$$
Our goal is to characterize $G$. We begin with two lemmas, where the proof of the first lemma is straightforward and hence omitted.

\setlength{\parindent}{0em}
\begin{lemma}
{Let $\varphi\in\End_{S}(\mathbf{H})$ and $\lambda\in\Aut(\mathbf{P})$. Then, the following three statements are equivalent to each other:

{\bf{(1)}}\,\,$\varphi\in\Aut_{S}(\mathbf{H})$, and $\langle\supp(\varphi(\alpha))\rangle_{\mathbf{P}}=\lambda[\langle\supp(\alpha)\rangle_{\mathbf{P}}]$ for all $\alpha\in\mathbf{H}$;

{\bf{(2)}}\,\,$\varphi\in\Aut_{S}(\mathbf{H})$. Moreover, for any $i\in K$ and $a\in H_{i}-\{0\}$, it holds that
$\langle\supp(\varphi(\eta_{i}(a)))\rangle_{\mathbf{P}}=\langle\{\lambda(i)\}\rangle_{\mathbf{P}}$;

{\bf{(3)}}\,\,$\pi_{j}\circ\varphi\circ\eta_{i}=0$ for all $i,j\in\Omega$ with $j\not\preccurlyeq_{\mathbf{P}}\lambda(i)$, and $\pi_{\lambda(i)}\circ\varphi\circ\eta_{i}\in\Aut_{S}(H_{i},H_{\lambda(i)})$ for all $i\in\Omega$.

}
\end{lemma}

\setlength{\parindent}{0em}
\begin{lemma}
{{\bf{(1)}}\,\,For $B\subseteq\Omega$ and $A\subseteq\langle B\rangle_{\mathbf{P}}$, it holds that $\varpi(A)\curlyeqprec\varpi(B)$.

{\bf{(2)}}\,\,For $B\subseteq\Omega$ and $v\in\langle B\rangle_{\mathbf{P}}$ such that $\varpi(B)\curlyeqprec\varpi(\{v\})$, it holds that $\langle B\rangle_{\mathbf{P}}=\langle \{v\}\rangle_{\mathbf{P}}$.

{\bf{(3)}}\,\,Let $\theta\in\mathbf{H}$ such that there exists $u\in\supp(\theta)$ with $\langle \supp(\theta)\rangle_{\mathbf{P}}=\langle \{u\}\rangle_{\mathbf{P}}$. Then, for $\gamma\in\mathbf{H}$, $\langle \supp(\gamma)\rangle_{\mathbf{P}}\subseteq\langle \supp(\theta)\rangle_{\mathbf{P}}$ if and only if both $\varpi(\supp(\gamma))\curlyeqprec\varpi(\supp(\theta))$ and $\varpi(\supp(\gamma+\theta))\curlyeqprec\varpi(\supp(\theta))$ hold true.

{\bf{(4)}}\,\,Let $\varphi\in\End_{S}(\mathbf{H})$ such that $\varpi(\supp(\varphi(\alpha)))=\varpi(\supp(\alpha))$ for all $\alpha\in\mathbf{H}$, and fix $\theta\in\mathbf{H}$ such that there exists $v\in\Omega$ with $\langle \supp(\varphi(\theta))\rangle_{\mathbf{P}}=\langle \{v\}\rangle_{\mathbf{P}}$. Then, there exists $u\in\supp(\theta)$ such that $\langle \supp(\theta)\rangle_{\mathbf{P}}=\langle \{u\}\rangle_{\mathbf{P}}$. Moreover, for any $\gamma\in\mathbf{H}$, we have $\langle \supp(\gamma)\rangle_{\mathbf{P}}\subseteq\langle \supp(\theta)\rangle_{\mathbf{P}}\Longleftrightarrow\langle \supp(\varphi(\gamma))\rangle_{\mathbf{P}}\subseteq\langle \supp(\varphi(\theta))\rangle_{\mathbf{P}}$.
}
\end{lemma}

\begin{proof}
{\bf{(1)}}\,\,By $\langle A\rangle_{\mathbf{P}}\subseteq\langle B\rangle_{\mathbf{P}}$ and (3.2), we have $\varpi(\langle A\rangle_{\mathbf{P}})\curlyeqprec\varpi(\langle B\rangle_{\mathbf{P}})$, which, along with (3.1), immediately implies that $\varpi(A)\curlyeqprec\varpi(B)$, as desired.

{\bf{(2)}}\,\,By (1), we have $\varpi(\{v\})\curlyeqprec\varpi(B)$. From $(Y,\curlyeqprec)$ is anti-symmetric and $\varpi(B)\curlyeqprec\varpi(\{v\})$, we infer that $\varpi(\{v\})=\varpi(B)$. By (3.1), we have $\varpi(\{v\})=\varpi(\langle B\rangle_{\mathbf{P}})$, which, along with (3.3), implies that $\langle B\rangle_{\mathbf{P}}=\langle \{v\}\rangle_{\mathbf{P}}$, as desired.

{\bf{(3)}}\,\,By (3.1), we have $\varpi(\supp(\theta))=\varpi(\{u\})$. Consider $\gamma\in\mathbf{H}$. The ``only if'' part can be readily derived by (1) and the fact that $\supp(\gamma+\theta)\subseteq\supp(\gamma)\cup\supp(\theta)$, and so we only prove the ``if'' part. If $u\in\supp(\gamma)$, then by $\varpi(\supp(\gamma))\curlyeqprec\varpi(\{u\})$ and (2), we have $\langle \supp(\gamma)\rangle_{\mathbf{P}}=\langle \{u\}\rangle_{\mathbf{P}}=\langle \supp(\theta)\rangle_{\mathbf{P}}$, as desired. Hence in the following, we assume that $u\not\in\supp(\gamma)$. Then, we have $u\in\supp(\gamma+\theta)$, which, along with $\varpi(\supp(\gamma+\theta))\curlyeqprec\varpi(\{u\})$ and (2), implies that $\langle \supp(\gamma+\theta)\rangle_{\mathbf{P}}=\langle \{u\}\rangle_{\mathbf{P}}=\langle \supp(\theta)\rangle_{\mathbf{P}}$. It follows that $\supp(\gamma)\subseteq\supp(\theta)\cup\supp(\gamma+\theta)\subseteq\langle\supp(\theta)\rangle_{\mathbf{P}}$, as desired.

{\bf{(4)}}\,\,We note that $\varphi(\theta)=\sum_{i\in\supp(\theta)}\varphi(\eta_{i}(\theta_{i}))$. Since $v\in\supp(\varphi(\theta))$, we can choose $u\in\supp(\theta)$ such that $v\in\supp(\varphi(\eta_{u}(\theta_{u})))$. It follows that $\varpi(\{v\})\curlyeqprec\varpi(\supp(\varphi(\eta_{u}(\theta_{u}))))=\varpi(\supp(\eta_{u}(\theta_{u})))=\varpi(\{u\})$. By (3.1), we have $\varpi(\{v\})=\varpi(\supp(\varphi(\theta)))=\varpi(\supp(\theta))$, which implies that $\varpi(\supp(\theta))\curlyeqprec\varpi(\{u\})$. Hence from (2), we have $\langle \supp(\theta)\rangle_{\mathbf{P}}=\langle \{u\}\rangle_{\mathbf{P}}$, as desired. Now for $\gamma\in\mathbf{H}$, since $\varpi(\supp(\gamma))=\varpi(\supp(\varphi(\gamma)))$, $\varpi(\supp(\theta))=\varpi(\supp(\varphi(\theta)))$, $\varpi(\supp(\gamma+\theta))=\varpi(\supp(\varphi(\gamma+\theta)))=\varpi(\supp(\varphi(\gamma)+\varphi(\theta)))$, the desired result follows from applying (3) to $(\gamma,\theta)$ and $(\varphi(\gamma),\varphi(\theta))$, respectively.
\end{proof}

\setlength{\parindent}{2em}
The following is the main result of this section.

\setlength{\parindent}{0em}
\begin{theorem}
{Assume that either $H_{i}\neq\{0\}$ for all $i\in\Omega$ or $\mathbf{P}$ is hierarchical. Then, we have:

{\bf{(1)}}\,\,For $\varphi\in G$, there uniquely exists $\lambda\in \Aut(\mathbf{P})$ such that $\lambda\mid_{\Omega-K}=\id_{\Omega-K}$ and $\langle\supp(\varphi(\alpha))\rangle_{\mathbf{P}}=\lambda[\langle\supp(\alpha)\rangle_{\mathbf{P}}]$ for all $\alpha\in\mathbf{H}$, and such a $\lambda$ is necessarily in $T$;

{\bf{(2)}}\,\,For $\psi\in\Aut_{S}(\mathbf{H})$ such that there exists $\mu\in T$ with $\langle\supp(\psi(\alpha))\rangle_{\mathbf{P}}=\mu[\langle\supp(\alpha)\rangle_{\mathbf{P}}]$ for all $\alpha\in\mathbf{H}$, we have $\psi\in G$;

{\bf{(3)}}\,\,There uniquely exists $\zeta:G\longrightarrow T$ such that for any $\varphi\in G$, it holds that $\langle\supp(\varphi(\alpha))\rangle_{\mathbf{P}}=\zeta_{(\varphi)}[\langle\supp(\alpha)\rangle_{\mathbf{P}}]$ for all $\alpha\in\mathbf{H}$. Moreover, we have $\zeta$ is a group homomorphism, $\ran(\zeta)=T$ and $\ker(\zeta)=\GL_{\mathbf{P}}(\mathbf{H})$.
}
\end{theorem}

\begin{proof}
{\bf{(1)}}\,\,First, consider an arbitrary $i\in K$. For any $a\in H_{i}-\{0\}$, applying (4) of Lemma 3.2 to $\varphi^{-1}\in G$ and $\varphi(\eta_{i}(a))\in\mathbf{H}$, we infer that there exists $u\in \supp(\varphi(\eta_{i}(a)))\subseteq K$ such that $\langle\supp(\varphi(\eta_{i}(a)))\rangle_{\mathbf{P}}=\langle\{u\}\rangle_{\mathbf{P}}$. Consider $b,c\in H_{i}-\{0\}$, and let $j,k\in K$ such that $\langle\supp(\varphi(\eta_{i}(b)))\rangle_{\mathbf{P}}=\langle\{j\}\rangle_{\mathbf{P}}$, $\langle\supp(\varphi(\eta_{i}(c)))\rangle_{\mathbf{P}}=\langle\{k\}\rangle_{\mathbf{P}}$. By (4) of Lemma 3.2, we have $k=j$. Hence we can fix $\sigma:K\longrightarrow K$ such that for any $i\in K$ and $a\in H_{i}-\{0\}$, it holds that
$\langle\supp(\varphi(\eta_{i}(a)))\rangle_{\mathbf{P}}=\langle\{\sigma(i)\}\rangle_{\mathbf{P}}$. We claim that $\sigma\in\Aut(K,\preccurlyeq_{\mathbf{P}})$. Indeed, let $i,t\in K$. Since $H_{i}\neq\{0\}$, $H_{t}\neq\{0\}$, we can choose $c\in H_{i}-\{0\}$, $d\in H_{t}-\{0\}$. Moreover, we have $\langle\supp(\varphi(\eta_{i}(c)))\rangle_{\mathbf{P}}=\langle\{\sigma(i)\}\rangle_{\mathbf{P}}$, $\langle\supp(\varphi(\eta_{t}(d)))\rangle_{\mathbf{P}}=\langle\{\sigma(t)\}\rangle_{\mathbf{P}}$. By (4) of Lemma 3.2, we have $\langle \{i\}\rangle_{\mathbf{P}}\subseteq\langle\{t\}\rangle_{\mathbf{P}}\Longleftrightarrow\langle \{\sigma(i)\}\rangle_{\mathbf{P}}\subseteq\langle\{\sigma(t)\}\rangle_{\mathbf{P}}$, which further implies that $i\preccurlyeq_{\mathbf{P}}t\Longleftrightarrow\sigma(i)\preccurlyeq_{\mathbf{P}}\sigma(t)$, as desired. Define $\lambda:\Omega\longrightarrow\Omega$ such that $\lambda\mid_{K}=\sigma$, $\lambda\mid_{\Omega-K}=\id\mid_{\Omega-K}$. Since either $K=\Omega$ or $\mathbf{P}$ is hierarchical holds true, we have $\lambda\in\Aut(\mathbf{P})$. Moreover, for any $i\in K$ and $a\in H_{i}-\{0\}$, it holds that
$\langle\supp(\varphi(\eta_{i}(a)))\rangle_{\mathbf{P}}=\langle\{\lambda(i)\}\rangle_{\mathbf{P}}$. By Lemma 3.1, we have $\langle\supp(\varphi(\alpha))\rangle_{\mathbf{P}}=\lambda[\langle\supp(\alpha)\rangle_{\mathbf{P}}]$ for all $\alpha\in\mathbf{H}$, and $H_{i}\cong H_{\lambda(i)}$ for all $i\in\Omega$. Now consider $I\subseteq K$. Then, we can choose $\alpha\in\mathbf{H}$ such that $\supp(\alpha)=I$. It follows that $\langle\supp(\varphi(\alpha))\rangle_{\mathbf{P}}=\langle\lambda[I]\rangle_{\mathbf{P}}$. By $\varphi\in G$ and (3.1), we have $\varpi(I)=\varpi(\supp(\alpha))=\varpi(\supp(\varphi(\alpha)))=\varpi(\lambda[I])$. The above discussion implies that $\lambda\in T$. Finally, we show the uniqueness of $\lambda$. Let $\mu\in\Aut(\mathbf{P})$ such that $\mu\mid_{\Omega-K}=\id_{\Omega-K}$ and $\langle\supp(\varphi(\alpha))\rangle_{\mathbf{P}}=\mu[\langle\supp(\alpha)\rangle_{\mathbf{P}}]$ for all $\alpha\in\mathbf{H}$. Let $t\in K$. Since $H_{t}\neq\{0\}$, we can choose $d\in H_{t}-\{0\}$. Note that $\langle\supp(\varphi(\eta_{t}(d)))\rangle_{\mathbf{P}}=\langle\{\lambda(t)\}\rangle_{\mathbf{P}}=\langle\{\mu(t)\}\rangle_{\mathbf{P}}$, we have $\lambda(t)=\mu(t)$. It immediately follows that $\lambda=\mu$, as desired.

{\bf{(2)}}\,\,This can be readily verified and hence we omit the details.

{\bf{(3)}}\,\,By (1), $\zeta$ is well defined and unique. A routine verification yields that $\zeta$ is a group homomorphism with $\ker(\zeta)=\GL_{\mathbf{P}}(\mathbf{H})$. Now we show that $\ran(\zeta)=T$. Consider $\mu\in T$. Then, we can choose $(\rho_{i}\mid i\in\Omega)\in\prod_{i\in\Omega}\Aut_{S}(H_{i},H_{\mu(i)})$. Define $\psi:\mathbf{H}\longrightarrow\mathbf{H}$ such that for any $\alpha\in \mathbf{H}$, $\psi(\alpha)_{\mu(i)}=\rho_{i}(\alpha_{i})$ for all $i\in\Omega$. It is straightforward to verify that $\psi\in\Aut_{S}(\mathbf{H})$ and $\langle\supp(\psi(\alpha))\rangle_{\mathbf{P}}=\mu[\langle\supp(\alpha)\rangle_{\mathbf{P}}]$ for all $\alpha\in\mathbf{H}$. From (2), we infer that $\psi\in G$, which further implies that $\mu=\zeta_{(\psi)}$, as desired.
\end{proof}

\setlength{\parindent}{2em}
Theorem 3.1 can be readily applied to weighted poset metric.

\setlength{\parindent}{0em}
\begin{corollary}
{Assume that either $H_{i}\neq\{0\}$ for all $i\in\Omega$ or $\mathbf{P}$ is hierarchical. Fix $\omega:\Omega\longrightarrow\mathbb{R}^{+}$, and define $Q\leqslant\Aut(\mathbf{P})$ as
$$\mbox{$Q=\{\mu\in\Aut(\mathbf{P})\mid\mu\mid_{\Omega-K}=\id_{\Omega-K},\text{$\omega(i)=\omega(\mu(i))$, $H_{i}\cong H_{\mu(i)}$ for all $i\in\Omega$}\}$}.$$
Then, for any $\varphi\in\GL_{(\mathbf{P},\omega)}(\mathbf{H})$, there uniquely exists $\lambda\in \Aut(\mathbf{P})$ such that $\lambda\mid_{\Omega-K}=\id_{\Omega-K}$ and $\langle\supp(\varphi(\alpha))\rangle_{\mathbf{P}}=\lambda[\langle\supp(\alpha)\rangle_{\mathbf{P}}]$ for all $\alpha\in\mathbf{H}$, and such a $\lambda$ is necessarily in $Q$. Conversely, for any $\psi\in\Aut_{S}(\mathbf{H})$ such that there exists $\mu\in Q$ with $\langle\supp(\psi(\alpha))\rangle_{\mathbf{P}}=\mu[\langle\supp(\alpha)\rangle_{\mathbf{P}}]$ for all $\alpha\in\mathbf{H}$, we have $\psi\in\GL_{(\mathbf{P},\omega)}(\mathbf{H})$. Moreover, there uniquely exists $\zeta:\GL_{(\mathbf{P},\omega)}(\mathbf{H})\longrightarrow Q$ such that for any $\varphi\in\GL_{(\mathbf{P},\omega)}(\mathbf{H})$, it holds that $\langle\supp(\varphi(\alpha))\rangle_{\mathbf{P}}=\zeta_{(\varphi)}[\langle\supp(\alpha)\rangle_{\mathbf{P}}]$ for all $\alpha\in\mathbf{H}$. In addition, we have $\zeta$ is a group homomorphism, $\ran(\zeta)=Q$ and $\ker(\zeta)=\GL_{\mathbf{P}}(\mathbf{H})$.
}
\end{corollary}

\begin{proof}
The result follows from applying Theorem 3.1 to $\varpi_1:2^{\Omega}\longrightarrow\mathbb{R}$ defined as $\varpi_1(B)=\sum_{i\in \langle B\rangle_{\mathbf{P}}}\omega(i)$.
\end{proof}

\begin{remark}
{In [35, Section II.A], Machado and Firer have studied a much more general function $\varpi:2^{\Omega}\longrightarrow\mathbb{R}$ that does not require (3.3), and their result [35, Theorem 1] applies to a wider range of metrics including combinatorial metric (see \cite{42}). On the other hand, their approach requires $\mathbf{H}$ to be a vector space over a non-binary field. Hence our Theorem 3.1 applies to more general module alphabets. In addition, if $\mathbf{H}$ is a vector space over a field, then Corollary 3.1 recovers [35, Theorem 5].
}
\end{remark}

\section{The MEP for $\mathbf{P}$-support}
\setlength{\parindent}{2em}
Throughout this section, we fix a poset $\mathbf{P}=(\Omega,\preccurlyeq_{\mathbf{P}})$. In addition, we let $m$ be the largest cardinality of a chain in $\mathbf{P}$, and for any $r\in[1,m]$, we define $W_{r}\triangleq\{u\in\Omega\mid \len_{\mathbf{P}}(u)=r\}$.

The following lemma gives some necessary conditions for the MEP.

\setlength{\parindent}{0em}
\begin{lemma}
{If for any linear code $C\subseteq\mathbf{H}$ and $\chi\in\Hom(C,\mathbf{H})$ such that $\chi$ preserves $\mathbf{P}$-support, there exists $\varphi\in\End_{S}(\mathbf{H})$ with $\varphi\mid_{C}=\chi$, then $(\mathbf{H},\mathbf{P})$ satisfies Condition (B). Furthermore, if $\mathbf{H}$ satisfies the MEP for $\mathbf{P}$-support, then $\mathbf{H}$ satisfies Condition (A).
}
\end{lemma}

\begin{proof}
Let $k,l\in\Omega$ such that $k\preccurlyeq_{\mathbf{P}}l$, $k\neq l$. Consider an $S$-submodule $B\subseteq H_{l}$ and $f\in\Hom_{S}(B,H_{k})$. Define $g\in\Hom_{S}(\eta_{l}[B],\mathbf{H})$ as $g(\eta_{l}(b))=\eta_{l}(b)+\eta_{k}(f(b))$ for all $b\in B$. We infer that $g$ preserves $\mathbf{P}$-support. Hence we can choose $\varphi\in\End_{S}(\mathbf{M})$ with $\varphi\mid_{\eta_{l}[B]}=g$. It follows that $\pi_{k}\circ\varphi\circ\eta_{l}$ is an element of $\Hom_{S}(H_{l},H_{k})$ which extends $f$, as desired. Now suppose that $\mathbf{H}$ satisfies the MEP for $\mathbf{P}$-support. Consider $i\in\Omega$. For an $S$-submodule $B\subseteq H_{i}$ and an injective $\xi\in\Hom_{S}(B,H_{i})$, we define $\chi\in\Hom(\eta_{i}[B],\mathbf{H})$ as $\chi(\eta_{i}(b))=\eta_{i}(\xi(b))$ for all $b\in B$. Since $\xi$ is injective, $\chi$ preserves $\mathbf{P}$-support. Hence we can choose $\varphi\in\GL_{\mathbf{P}}(\mathbf{H})$ with $\varphi\mid_{\eta_{i}[B]}=\chi$. From Lemma 3.1 and a routine verification, we deduce that $\pi_{i}\circ\varphi\circ\eta_{i}\in\Aut_{S}(H_{i})$ and $(\pi_{i}\circ\varphi\circ\eta_{i})\mid_{B}=\xi$, as desired.
\end{proof}

\setlength{\parindent}{2em}
Now we show that if $\mathbf{P}$ is hierarchical, then the converse of the second part of Lemma 4.1 holds true as well, which also leads to a canonical decomposition for semi-simple linear codes. The following theorem is the first main result of this section.

\setlength{\parindent}{0em}
\begin{theorem}
{Assume that $\mathbf{P}$ is hierarchical, $\mathbf{H}$ satisfies Condition (A) and $(\mathbf{H},\mathbf{P})$ satisfies Condition (B). Then, we have:

{\bf{(1)}}\,\,Consider $r\in[1,m]$. Let $C\subseteq\delta(\bigcup_{j=1}^{r}W_{j})$ be a linear code, and let $f\in\Hom_{S}(C,\mathbf{H})$ such that $f$ preserves $\mathbf{P}$-support. Then, there exists $\varphi\in\GL_{\mathbf{P}}(\mathbf{H})$ such that $\varphi\mid_{C}=f$ and $\varphi(\alpha)=\alpha$ for all $\alpha\in \delta(\bigcup_{j=r+1}^{m}W_{j})$;

{\bf{(2)}}\,\,$\mathbf{H}$ satisfies the MEP for $\mathbf{P}$-support;

{\bf{(3)}}\,\,Consider $r\in[1,m]$. Let $C\subseteq\delta(\bigcup_{j=1}^{r}W_{j})$ be a semi-simple linear code. Then, we have $\varphi[C]=B_{1}+\cdots+B_{r}$ for some $\varphi\in\GL_{\mathbf{P}}(\mathbf{H})$ such that $\varphi(\alpha)=\alpha$ for all $\alpha\in \delta(\bigcup_{j=r+1}^{m}W_{j})$, and $B_{j}\subseteq\delta(W_{j})$ for all $j\in[1,r]$.
}
\end{theorem}

\begin{proof}
{\bf{(1)}}\,\,Throughout the proof, for any $h_1,h_2\in\Hom_{S}(C,\mathbf{H})$, we write $h_1\equiv h_2$ if there exists $\sigma\in\GL_{\mathbf{P}}(\mathbf{H})$ such that $\sigma(\alpha)=\alpha$ for all $\alpha\in \delta(\bigcup_{j=r+1}^{m}W_{j})$ and $h_2=\sigma\circ h_1$. We also write $D\triangleq C\cap\delta(\bigcup_{j=1}^{r-1}W_{j})$.

\hspace*{4mm}\,\,First, we show that there exists $g_1\in\Hom_{S}(C,\mathbf{H})$ such that $f\equiv g_1$, ${g_1}\mid_{D}=\id_{D}$ and for any $\alpha\in C$, $i\in W_{r}$, it holds that $g_1(\alpha)_{i}=\alpha_{i}$. Applying an induction argument to $r-1$, $D$ and $f\mid_{D}$, we can choose $\tau\in\GL_{\mathbf{P}}(\mathbf{H})$ such that $\tau\mid_{D}=f\mid_{D}$ and $\tau(\alpha)=\alpha$ for all $\alpha\in \delta(\bigcup_{j=r}^{m}W_{j})$. Define $g\triangleq\tau^{-1}\circ f$. Consider an arbitrary $i\in W_{r}$. We claim that there exists $\varsigma\in\Aut_{S}(H_{i})$ such that $\varsigma(g(\alpha)_{i})=\alpha_{i}$ for all $\alpha\in C$. Indeed, consider $\alpha\in C$. Since $\langle\supp(g(\alpha))\rangle_{\mathbf{P}}=\langle\supp(\alpha)\rangle_{\mathbf{P}}\subseteq\bigcup_{j=1}^{r}W_{j}$, we have $i\in\supp(\alpha)\Longleftrightarrow i\in\supp(g(\alpha))$, which implies that $\alpha\in\ker(\pi_{i}\mid_{C})\Longleftrightarrow\alpha\in\ker(\pi_{i}\circ g)$. It follows that $\ker(\pi_{i}\mid_{C})=\ker(\pi_{i}\circ g)$. Hence there exists an injective map $\varsigma_1\in\Hom_{S}((\pi_{i}\circ g)[C],H_{i})$ with $\varsigma_1\circ\pi_{i}\circ g=\pi_{i}\mid_{C}$. Since $H_{i}$ is strong pseudo-injective, we can choose $\varsigma\in\Aut_{S}(H_{i})$ with $\varsigma\mid_{(\pi_{i}\circ g)[C]}=\varsigma_1$. Apparently, we have $\varsigma\circ\pi_{i}\circ g=\pi_{i}\mid_{C}$, and hence $\varsigma(g(\alpha)_{i})=\alpha_{i}$ for all $\alpha\in C$, as desired. Therefore we can choose $(\mu_{i}\mid i\in W_{r})$ such that for any $i\in W_{r}$, it holds that $\mu_{i}\in\Aut_{S}(H_{i})$ and $\mu_{i}(g(\alpha)_{i})=\alpha_{i}$ for all $\alpha\in C$. Now there uniquely exists $\psi\in\End_{S}(\mathbf{H})$ such that for any $\alpha\in\mathbf{H}$, $\psi(\alpha)\in\mathbf{H}$ is defined as $(\forall~i\in W_{r}:\psi(\alpha)_{i}=\mu_{i}(\alpha_{i}))$ and $(\forall~t\in \Omega-W_{r}:\psi(\alpha)_{t}=\alpha_{t})$. Apparently, we have $\psi\in\GL_{\mathbf{P}}(\mathbf{H})$ and $\psi(\alpha)=\alpha$ for all $\alpha\in \delta(\Omega-W_{r})$. Define $g_1\triangleq\psi\circ g$. Then, one can check that $f\equiv g_1$, $g_1\mid_{D}=\id_{D}$. Moreover, for $i\in W_{r}$ and $\alpha\in C$, we have $g_1(\alpha)_{i}=\psi(g(\alpha))_{i}=\mu_{i}(g(\alpha)_{i})=\alpha_{i}$, as desired.

\hspace*{4mm}\,\,Next, define $h:C\longrightarrow\mathbf{H}$ as $h(\alpha)=g_1(\alpha)-\alpha$. It can be readily verified that $h\in\Hom_{S}(C,\delta(\bigcup_{j=1}^{r-1}W_{j}))$ and $D\subseteq\ker(h)$. Hence there uniquely exists $\rho\in\Hom_{S}(C/D,\delta(\bigcup_{j=1}^{r-1}W_{j}))$ such that $\rho(\alpha+D)=h(\alpha)$ for all $\alpha\in C$. Since $C\subseteq\delta(\bigcup_{j=1}^{r}W_{j})$ and $D=C\cap\delta(\bigcup_{j=1}^{r-1}W_{j})$, we can define an injective map $\varepsilon\in\Hom_{S}(C/D,\delta(W_{r}))$ such that $\varepsilon(\alpha+D)=\sum_{i\in W_{r}}\eta_{i}(\alpha_{i})$ for all $\alpha\in C$. From $\mathbf{P}$ is hierarchical, we infer that $H_{k}$ is $H_{l}$-injective for all $k,l\in\Omega$ with $\len_{\mathbf{P}}(k)+1\leqslant\len_{\mathbf{P}}(l)$, which further implies that $\delta(\bigcup_{j=1}^{r-1}W_{j})$ is $\delta(W_{r})$-injective. Since $\varepsilon$ is injective, we can choose $\lambda\in\Hom_{S}(\delta(W_{r}),\delta(\bigcup_{j=1}^{r-1}W_{j}))$ with $\rho=\lambda\circ\varepsilon$. From the definition of $\rho$ and $\varepsilon$, we deduce that $g_1(\alpha)=\alpha+\lambda(\sum_{i\in W_{r}}\eta_{i}(\alpha_{i}))$ for all $\alpha\in C$.
Now define $\sigma\in\End_{S}(\mathbf{H})$ as $\sigma(\gamma)=\gamma+\lambda(\sum_{i\in W_{r}}\eta_{i}(\gamma_{i}))$. It is straightforward to verify that for any $l\in\Omega$, $\pi_{l}\circ\sigma\circ\eta_{l}=\id_{H_{l}}$. Moreover, for any $k,l\in\Omega$ such that $k\neq l$ and $\pi_{k}\circ\sigma\circ\eta_{l}\neq0$, we have $k\in\bigcup_{j=1}^{r-1}W_{j}$ and $l\in W_{r}$, which, along with $\mathbf{P}$ is hierarchical, implies that $k\preccurlyeq_{\mathbf{P}}l$. Now Lemma 3.1 implies that $\sigma\in\GL_{\mathbf{P}}(\mathbf{H})$. In addition, by the definition of $\sigma$, we have $\sigma\mid_{C}=g_1$ and $\sigma(\gamma)=\gamma$ for all $\gamma\in \delta(\Omega-W_{r})$. Finally, by $f\equiv g_1$, we conclude that (1) holds true, as desired.

{\bf{(2)}}\,\,This immediately follows from (1).

{\bf{(3)}}\,\,Let $D\triangleq C\cap\delta(\bigcup_{j=1}^{r-1}W_{j})$. Since $C$ is semi-simple, we can choose an $S$-submodule $L$ of $C$ such that $C=D+L$, $D\cap L=\delta(\bigcup_{j=1}^{r-1}W_{j})\cap L=\{0\}$.
Define $f\in\Hom_{S}(L,\delta(W_{r}))$ as $f(\beta)=\sum_{i\in W_{r}}\eta_{i}(\beta_{i})$. Since $L\subseteq\delta(\bigcup_{j=1}^{r}W_{j})$, $\delta(\bigcup_{j=1}^{r-1}W_{j})\cap L=\{0\}$ and $\mathbf{P}$ is hierarchical, $f$ preserves $\mathbf{P}$-support. By (1), we can choose $\sigma\in\GL_{\mathbf{P}}(\mathbf{H})$ such that $\sigma\mid_{L}=f$ and $\sigma(\alpha)=\alpha$ for all $\alpha\in \delta(\bigcup_{j=r+1}^{m}W_{j})$. We infer that $\sigma[C]=\sigma[D]+f[L]$. Noticing that $\sigma[D]\subseteq\delta(\bigcup_{j=1}^{r-1}W_{j})$ and $\sigma[D]$ is semi-simple, applying an induction argument to $r-1$, we have $\psi[\sigma[D]]=E_{1}+\cdots+E_{r-1}$ for some $\psi\in\GL_{\mathbf{P}}(\mathbf{H})$ such that $\psi(\alpha)=\alpha$ for all $\alpha\in \delta(\bigcup_{j=r}^{m}W_{j})$, and $E_{j}\subseteq\delta(W_{j})$ for all $j\in[1,r-1]$. Since $f[L]\subseteq\delta(W_{r})$, we have $\psi[f[L]]=f[L]$, which further implies that $(\psi\circ\sigma)[C]=E_{1}+\cdots+E_{r-1}+f[L]$, as desired.
\end{proof}

\setlength{\parindent}{0em}
\begin{remark}
{Part (3) of Theorem 4.1 can be regarded as a canonical decomposition for semi-simple codes. It generalizes [19, Corollary 1] and the ``only if'' parts of [18, Theorem 9], [33, Theorem 1], [34, Theorem 2] and [35, Theorem 6] to codes over modules.
}
\end{remark}

\setlength{\parindent}{2em}
Now we consider some other sufficient conditions for $\mathbf{H}$ to satisfy the MEP for $\mathbf{P}$-support. From now on, we let $R$ denote the following set
$$\{\varphi\in\End_{S}(\mathbf{H})\mid\text{$\supp(\varphi(\alpha))\subseteq\langle\supp(\alpha)\rangle_{\mathbf{P}}$ for all $\alpha\in\mathbf{H}$}\}.$$
We note that $R$ is a subring of $\End_{S}(\mathbf{H})$ with $\id_{\mathbf{H}}\in R$, and $\GL_{\mathbf{P}}(\mathbf{H})$ is exactly the set of all the multiplicative invertible elements of $R$.

\setlength{\parindent}{0em}
\begin{lemma}
{The following three statements are equivalent to each other:

{\bf{(1)}}\,\,For any linear code $C\subseteq\mathbf{H}$ and $\chi\in\Hom_{S}(C,\mathbf{H})$ such that $\supp(\chi(\alpha))\subseteq\langle\supp(\alpha)\rangle_{\mathbf{P}}$ for all $\alpha\in C$, there exists $\varphi\in R$ with $\varphi\mid_{C}=\chi$;

{\bf{(2)}}\,\,For any linear code $C\subseteq\mathbf{H}$ and $\chi\in\Hom_{S}(C,\mathbf{H})$ such that $\supp(\chi(\alpha))\subseteq\langle\supp(\alpha)\rangle_{\mathbf{P}}$ for all $\alpha\in C$, there exists $\varphi\in\End_{S}(\mathbf{H})$ with $\varphi\mid_{C}=\chi$;

{\bf{(3)}}\,\,For any $k,l\in\Omega$ with $k\preccurlyeq_{\mathbf{P}}l$, $H_{k}$ is $H_{l}$-injective.
}
\end{lemma}

\begin{proof}
We note that $(1)\Longrightarrow(2)$ is trivial and $(2)\Longrightarrow(3)$ can be proven similarly as in the proof of Lemma 4.1, and so we only prove $(3)\Longrightarrow(1)$. Let $C\subseteq\mathbf{H}$ be a linear code and let $\chi\in\Hom_{S}(C,\mathbf{H})$ such that $\supp(\chi(\alpha))\subseteq\langle\supp(\alpha)\rangle_{\mathbf{P}}$ for all $\alpha\in C$. We define a tuple $(\rho_{(i,k)}\mid (i,k)\in\Omega\times\Omega)\in\prod_{(i,k)\in\Omega\times\Omega}\Hom_{S}(H_{i},H_{k})$ as follows. Consider a fixed $k\in\Omega$, and let $E=\langle\{k\}\rangle_{\mathbf{\overline{P}}}$. Then, for $\pi_{k}\circ\chi\in\Hom(C,H_{k})$ and $\zeta\in\Hom(C,\delta(E))$ defined as $\zeta(\alpha)=\sum_{i\in E}\eta_{i}(\alpha_{i})$ for all $\alpha\in C$, we have $\ker(\zeta)\subseteq\ker(\pi_{k}\circ\chi)$. Hence there uniquely exists $\lambda\in\Hom_{S}(\zeta[C],H_{k})$ with $\lambda\circ\zeta=\pi_{k}\circ\chi$. Now for any $l\in E$, by $k\preccurlyeq_{\mathbf{P}}l$, we have $H_{k}$ is $H_{l}$-injective. It follows that $H_{k}$ is $\delta(E)$-injective. Hence we can choose $\mu\in\Hom(\delta(E),H_{k})$ with $\mu\mid_{\zeta[C]}=\lambda$. Moreover, we set $(\rho_{(i,k)}\mid i\in\Omega)\in\prod_{i\in\Omega}\Hom_{S}(H_{i},H_{k})$ as $\rho_{(i,k)}=\mu\circ\eta_{i}$ for all $i\in\langle\{k\}\rangle_{\mathbf{\overline{P}}}$, and $\rho_{(i,k)}=0$ for all $i\in\Omega-\langle\{k\}\rangle_{\mathbf{\overline{P}}}$. Consider $\varphi\in\End_{S}(\mathbf{H})$ defined as $\pi_{k}\circ\varphi\circ\eta_{i}=\rho_{(i,k)}$ for all $(i,k)\in\Omega\times\Omega$. It is straightforward to verify that $\varphi\in R$ and $\varphi\mid_{C}=\chi$, as desired.
\end{proof}

\setlength{\parindent}{2em}
The following is the second main result of this section.

\setlength{\parindent}{0em}
\begin{theorem}
{Suppose that $R$ is semi-local, i.e., $R/\Jac(R)$ is a semi-simple ring, where $\Jac(R)$ is the Jacobson radical of $R$. Further assume that $H_{k}$ is $H_{l}$-injective for all $k,l\in\Omega$ with $k\preccurlyeq_{\mathbf{P}}l$. Then, $\mathbf{H}$ satisfies the MEP for $\mathbf{P}$-support.
}
\end{theorem}

\begin{proof}
Let $C\subseteq\mathbf{H}$ be a linear code. Then, $\Hom_{S}(C,\mathbf{H})$ is a left $R$-module via composition of maps. Let $f\in\Hom_{S}(C,\mathbf{H})$ such that $f$ preserves $\mathbf{P}$-support. By Lemma 4.2, we can choose $\sigma\in R$ such that $f=\sigma\mid_{C}=\sigma\circ\id_{C}$. Noticing that $f$ is injective and $f^{-1}\in\Hom_{S}(f[C],\mathbf{H})$ also preserves $\mathbf{P}$-support, by Lemma 4.2, we can choose $\tau\in R$ such that $\tau\mid_{f[C]}=f^{-1}$. It follows that $\tau\circ f=\id_{C}$. Since $R$ is semi-local and $\GL_{\mathbf{P}}(\mathbf{H})$ is exactly the set of all the multiplicative invertible elements of $R$, from [4, Lemma 6.4] (also see [47, Proposition 5.1]), we conclude that there exists $\varphi\in\GL_{\mathbf{P}}(\mathbf{H})$ such that $f=\varphi\circ\id_{C}=\varphi\mid_{C}$, as desired.
\end{proof}

\setlength{\parindent}{2em}
We end this section by giving some consequences of Theorems 4.1 and 4.2. In the following corollary, we apply Theorem 4.1 to Rosenbloom-Tsfasman weight (see \cite{3,21,43}), i.e., poset weight induced by a chain, and apply Theorem 4.2 to some specific alphabets.

\setlength{\parindent}{0em}
\begin{corollary}
{{\bf{(1)}}\,\,Suppose that $\mathbf{P}$ is a chain, and fix $\omega:\Omega\longrightarrow\mathbb{R}^{+}$. Then, $\mathbf{H}$ satisfies the MEP for $(\mathbf{P},\omega)$-weight if and only if $\mathbf{H}$ satisfies Condition (A) and $(\mathbf{H},\mathbf{P})$ satisfies Condition (B).

{\bf{(2)}}\,\,Assume that either $\mathbf{H}$ is finite and $H_{k}$ is $H_{l}$-injective for all $k,l\in\Omega$ with $k\preccurlyeq_{\mathbf{P}}l$, or $S$ is a division ring and $\mathbf{H}$ is a finite dimensional left $S$-module. Then, $\mathbf{H}$ satisfies the MEP for $\mathbf{P}$-support.
}
\end{corollary}

\begin{proof}
{\bf{(1)}}\,\,Since $\mathbf{P}$ is a chain, we infer that for any $\alpha,\beta\in\mathbf{H}$, $\wt_{(\mathbf{P},\omega)}(\alpha)=\wt_{(\mathbf{P},\omega)}(\beta)\Longleftrightarrow\langle\supp(\alpha)\rangle_{\mathbf{P}}=\langle\supp(\beta)\rangle_{\mathbf{P}}$. It follows that $\mathbf{H}$ satisfies the MEP for $(\mathbf{P},\omega)$-weight if and only if $\mathbf{H}$ satisfies the MEP for $\mathbf{P}$-support. Noticing that $\mathbf{P}$ is hierarchical, the desired result immediately follows from Lemma 4.1 and Theorem 4.1.

{\bf{(2)}}\,\,The assumption ensures that $R$ is semi-local and $H_{k}$ is $H_{l}$-injective for all $k,l\in\Omega$ with $k\preccurlyeq_{\mathbf{P}}l$, and hence Theorem 4.2 concludes the proof.
\end{proof}

\setlength{\parindent}{0em}
\begin{remark}
{Part (1) of Corollary 4.1 generalizes [3, Theorem 6.1], [21, Theorem 4.13] and [21, Theorem 5.1], and if $\mathbf{P}$ is an anti-chain, then (2) of Corollary 4.1 recovers [3, Theorem 6.3] and [21, Remark 4.21 (a)]. All the aforementioned results have been established for codes over finite Frobenius rings and finite Frobenius bimodules.
}
\end{remark}

\section{The isometry equation}
\setlength{\parindent}{2em}
Isometry equation has been introduced by Dyshko to study the MEP for Hamming weight and symmetrized weight composition (see \cite{13}--\cite{17} and [29, Lemma 4.4]), and we first recall some basic facts. For any set $X$ and $Y\subseteq X$, the \textit{indicator function} $\mathbf{1}_{(X,Y)}:X\longrightarrow\{0,1\}$ is defined as
\begin{equation}\mbox{$\mathbf{1}_{(X,Y)}(x)=1$ if and only if $x\in Y$}.\end{equation}
Let $I$ and $J$ be finite sets, and let $(A_{i}\mid i\in I)$, $(B_{j}\mid j\in J)$ be two tuple of sets. We say that $\mathcal{U}\triangleq((A_{i}\mid i\in I),(B_{j}\mid j\in J))$ is \textit{a solution to the isometry equation}, or simply \textit{a solution}, if for some set $X$ such that $(\bigcup_{i\in I}A_{i})\cup(\bigcup_{j\in J}B_{i})\subseteq X$, the following \textit{isometry equation} holds:
\begin{equation}\sum_{i\in I}\mathbf{1}_{(X,A_{i})}=\sum_{j\in J}\mathbf{1}_{(X,B_{j})}.\end{equation}
Further assume that $\mathcal{U}$ is a solution. Then, one can check that for any set $X$ with $(\bigcup_{i\in I}A_{i})\cup(\bigcup_{j\in J}B_{i})\subseteq X$, the isometry equation (5.2) holds true, and for any set $C$, $((A_{i}\cap C\mid i\in I),(B_{j}\cap C\mid j\in J))$ is a solution. In addition, $\mathcal{U}$ is said to be \textit{trivial} if there exists a bijection $\sigma:I\longrightarrow J$ such that $A_{i}=B_{\sigma(i)}$ for all $i\in I$, and is said to be \textit{nontrivial} otherwise.

Dyshko has given the connection between the MEP for Hamming weight and the isometry equation. He has also determined the minimal length of nontrivial solutions with respect to the submodule lattice of a matrix module, and has established necessary and sufficient conditions for a matrix module alphabet to satisfy the MEP for Hamming weight. We first collect some of his results in \cite{14}, as detailed in the following three lemmas.

\setlength{\parindent}{0em}
\begin{lemma}([14, Proposition 1])
{Let $M$ be a  left $S$-module and let $n\in\mathbb{Z}^{+}$. For any $i\in [1,n]$, define $\varepsilon_{i}:M^{n}\longrightarrow M$ as $\varepsilon_{i}(\alpha)=\alpha_{i}$. Consider a linear code $C\subseteq M^{n}$ and $f\in\Hom_{S}(C,M^{n})$, and let $\mathcal{U}=((\ker(\varepsilon_{i})\cap C\mid i\in[1,n]),(\ker(\varepsilon_{i}\circ f)\mid i\in[1,n]))$. Then, $f$ preserves Hamming weight if and only if $\mathcal{U}$ is a solution. If $f$ extends to a Hamming weight isometry of $M^{n}$, then $\mathcal{U}$ is a trivial solution. Conversely, if $M$ is strong pseudo-injective and $\mathcal{U}$ is a trivial solution, then $f$ extends to a Hamming weight isometry of $M^{n}$.
}
\end{lemma}

\setlength{\parindent}{0em}
\begin{lemma}([14, Lemma 6])
{Let $\mathbb{F}$ be a finite field, $e\in\mathbb{Z}^{+}$, and suppose that $S=Mat_{e}(\mathbb{F})$. Moreover, let $X$ be a finite left $S$-module, $n\in\mathbb{Z}^{+}$, and let $(U_{1},\dots,U_{n})$, $(V_{1},\dots,V_{n})$ be two tuples of $S$-submodules of $X$ such that $((U_{1},\dots,U_{n}),(V_{1},\dots,V_{n}))$ is a nontrivial solution. Then, it holds true that $n\geqslant\prod_{i=1}^{e}(|\mathbb{F}|^{i}+1)$.
}
\end{lemma}

\setlength{\parindent}{0em}
\begin{lemma}([14, Propositions 2, 3, along with a special case of Theorem 3])
{Let $\mathbb{F}$ be a finite field, $e\in\mathbb{Z}^{+}$, and suppose that $S=Mat_{e}(\mathbb{F})$. Then, for $k,n\in\mathbb{Z}^{+}$, the left $S$-module $Mat_{e,k}(\mathbb{F})^{n}$ satisfies the MEP for Hamming weight if and only if either $k\leqslant e$ or $n\leqslant(\prod_{i=1}^{e}(|\mathbb{F}|^{i}+1))-1$ holds true.
}
\end{lemma}

\setlength{\parindent}{2em}
Our main goal of this section is to generalize Lemma 5.2 to a wider range of lattices. More precisely, we let $X$ be a set, and let $\Gamma$ be a finite subset of $2^{X}$ such that $X\in\Gamma$ and $\Gamma$ is closed under intersection. By [45, Proposition 3.3.1], $(\Gamma,\subseteq)$ is a lattice. For any $A\subseteq X$, we write $\langle A\rangle_{\Gamma}\triangleq\bigcap_{(D\in\Gamma,A\subseteq D)}D$. Let $\mu:\Gamma\times\Gamma\longrightarrow\mathbb{Z}$ be the M\"{o}bius function of $(\Gamma,\subseteq)$. Following [45, Section 3.7], $\mu$ can be characterized by the following three properties:

\setlength{\parindent}{0em}
{\bf{(a)}}\,\,$\mu(A,B)=0$ for all $A,B\in\Gamma$ with $A\nsubseteq B$;

{\bf{(b)}}\,\,$\mu(C,C)=1$ for all $C\in\Gamma$;

{\bf{(c)}}\,\,$\sum_{(U\in\Gamma,C\subseteq U\subseteq D)}\mu(U,D)=0$ for all $C,D\in\Gamma$ with $C\neq D$.\\

\setlength{\parindent}{2em}
With the help of Properties (a), (b), (c), the proof of the following lemma is straightforward and hence omitted.

\setlength{\parindent}{0em}
\begin{lemma}
{{\bf{(1)}}\,\,$\bigcap_{D\in\Gamma}D\in\Gamma$. Moreover, for $C\in\Gamma$ such that $C\neq \bigcap_{D\in\Gamma}D$, it holds that $\sum_{(U\in\Gamma,U\subseteq C)}\mu(U,C)=0$ and
\begin{eqnarray*}
\begin{split}
\sum_{(U\in\Gamma,\mu(U,C)\geqslant1)}\mu(U,C)=\sum_{(U\in\Gamma,\mu(U,C)\leqslant-1)}-\mu(U,C)=\frac{1}{2}\left(\sum_{(U\in\Gamma,U\subseteq C)}|\mu(U,C)|\right).
\end{split}
\end{eqnarray*}

{\bf{(2)}}\,\,Let $Y\in\Gamma$, and let $E=\{x\in X\mid\langle \{x\}\rangle_{\Gamma}=Y\}$. Then, we have
\begin{equation}\sum_{(U\in\Gamma,U\subseteq Y)}\mu(U,Y)\cdot\mathbf{1}_{(X,U)}=\mathbf{1}_{(X,E)}.\end{equation}
Moreover, $Y\neq\langle\{x\}\rangle_{\Gamma}$ for all $x\in X$ if and only if it holds that
\begin{equation}\sum_{(U\in\Gamma,\mu(U,Y)\leqslant-1)}-\mu(U,Y)\cdot\mathbf{1}_{(X,U)}=\sum_{(U\in\Gamma,\mu(U,Y)\geqslant1)}\mu(U,Y)\cdot\mathbf{1}_{(X,U)}.\end{equation}

{\bf{(3)}}\,\,Suppose that $\emptyset\not\in\Gamma$. Let $I$ and $J$ be finite sets, and let $((U_{i}\mid i\in I),(V_{j}\mid j\in J))\in\Gamma^{I}\times\Gamma^{J}$ be a solution. Then, it holds that $|I|=|J|$.
}
\end{lemma}

\setlength{\parindent}{2em}
Now we characterize nontrivial solutions that satisfy certain minimal condition.

\setlength{\parindent}{0em}
\begin{proposition}
{Suppose that $\emptyset\not\in\Gamma$. Let $K$ and $L$ be finite sets, and fix a non-trivial solution $((U_{i}\mid i\in K),(V_{j}\mid j\in L))\in\Gamma^{K}\times\Gamma^{L}$. Assume that for any $I\subsetneqq K$, $J\subsetneqq L$, $((U_{i}\mid i\in I),(V_{j}\mid j\in J))$ is not a non-trivial solution, and for any $W\in\Gamma$ such that $(\bigcup_{i\in K}U_{i})\cup(\bigcup_{j\in L}V_{j})\nsubseteq W$, $((U_{i}\cap W\mid i\in K),(V_{j}\cap W\mid j\in L))$ is a trivial solution. Then, there uniquely exists $Y\in\{U_{i}\mid i\in K\}\cup\{V_{j}\mid j\in L\}$ such that $U_{i},V_{j}\subseteq Y$ for all $i\in K$, $j\in L$. Moreover, it holds that:

{\bf{(1)}}\,\,$Y\neq\langle\{x\}\rangle_{\Gamma}$ for all $x\in X$;

{\bf{(2)}}\,\,$|K|=|L|=\frac{1}{2}(\sum_{(W\in\Gamma,W\subseteq Y)}|\mu(W,Y)|)$;

{\bf{(3)}}\,\,Assume that $Y\in\{V_{j}\mid j\in L\}$. Then, it holds that $\{U_{i}\mid i\in K\}=\{C\in\Gamma\mid \mu(C,Y)\leqslant-1\}$, $\{V_{j}\mid j\in L\}=\{D\in\Gamma\mid \mu(D,Y)\geqslant1\}$, $|\{i\in K\mid U_{i}=C\}|=-\mu(C,Y)$ for all $C\in\Gamma$ with $\mu(C,Y)\leqslant-1$, and $|\{j\in L\mid V_{j}=D\}|=\mu(D,Y)$ for all $D\in\Gamma$ with $\mu(D,Y)\geqslant1$.
}
\end{proposition}

\begin{proof}
Throughout the proof, we let $\Delta=\{U_{i}\mid i\in K\}\cup\{V_{j}\mid j\in L\}$. By Lemma 5.4, we have $|K|=|L|\geqslant1$. Consider $r\in K$, $t\in L$. Suppose that $U_{r}=V_{t}$. Then, $((U_{i}\mid i\in K-\{r\}),(V_{j}\mid j\in L-\{t\}))$ is a solution, which is necessarily trivial since $K-\{r\}\subsetneqq K$, $L-\{t\}\subsetneqq L$. Along with $U_{r}=V_{t}$, we deduce that $((U_{i}\mid i\in K),(V_{j}\mid j\in L))$ is a trivial solution, a contradiction. It follows that $U_{r}\neq V_{t}$, which further implies that
\begin{equation}\{U_{i}\mid i\in K\}\cap\{V_{j}\mid j\in L\}=\emptyset.\end{equation}
Now let $Y$ be a maximal element of $\Delta$. Without loss of generality, we assume that $Y\in\{V_{j}\mid j\in L\}$. Suppose that $(\bigcup_{i\in K}U_{i})\cup(\bigcup_{j\in L}V_{j})\nsubseteq Y$. Then, $((U_{i}\cap Y\mid i\in K),(V_{j}\cap Y\mid j\in L))$ is a trivial solution. In particular, we can choose $r\in K$, $t\in L$ with $U_{r}\cap Y=V_{t}=Y$. It follows from the maximality of $Y$ that $U_{r}=Y=V_{t}$, a contradiction to (5.5). Hence $Y$ is the greatest element of $\Delta$. Now define $f,g,\varphi,\psi:\Gamma\longrightarrow\mathbb{N}$ as $f(C)=|\{i\in K\mid U_{i}=C\}|$, $g(C)=|\{i\in K\mid C\subseteq U_{i}\}|$, $\varphi(C)=|\{j\in L\mid V_{j}=C\}|$, $\psi(C)=|\{j\in L\mid C\subseteq V_{j}\}|$. We note that $f(Y)=g(Y)=0$, $\varphi(Y)=\psi(Y)\triangleq e\in\mathbb{Z}^{+}$, and $g(Q)=\psi(Q)=0$ for all $Q\in\Gamma$ with $Q\nsubseteq Y$. Consider an arbitrary $D\in\Gamma$ with $D\subseteq Y$. By the M\"{o}bius inversion formula (see [45, Proposition 3.7.2]), we have
$$f(D)=\sum_{(E\in\Gamma,D\subseteq E\subseteq Y)}\mu(D,E)g(E),~\varphi(D)=\sum_{(E\in\Gamma,D\subseteq E\subseteq Y)}\mu(D,E)\psi(E).$$
For $E\in\Gamma$ with $D\subseteq E\subsetneqq Y$, since $((U_{i}\cap E\mid i\in K),(V_{j}\cap E\mid j\in L))$ is a trivial solution, we have $|\{i\in K\mid U_{i}\cap E=E\}|=|\{j\in L\mid V_{j}\cap E=E\}|$, which further implies that $g(E)=\psi(E)$. The above discussion yields that $\varphi(D)-f(D)=\mu(D,Y)\cdot e$. In addition, from (5.5), we infer that $0\in\{f(D),\varphi(D)\}$. Hence via some straightforward verification, we deduce that $\{U_{i}\mid i\in K\}=\{C\in\Gamma\mid \mu(C,Y)\leqslant-1\}$, $f(C)=-\mu(C,Y)\cdot e$ for all $C\in\Gamma$ with $\mu(C,Y)\leqslant-1$; $\{V_{j}\mid j\in L\}=\{D\in\Gamma\mid \mu(D,Y)\geqslant1\}$, and $\varphi(D)=\mu(D,Y)\cdot e$ for all $D\in\Gamma$ with $\mu(D,Y)\geqslant1$. Noticing that $\sum_{i\in K}\mathbf{1}_{(X,U_{i})}=\sum_{j\in L}\mathbf{1}_{(X,V_{j})}$, we have
\begin{equation}\sum_{(C\in\Gamma,\mu(C,Y)\leqslant-1)}-\mu(C,Y)\cdot\mathbf{1}_{(X,C)}=\sum_{(D\in\Gamma,\mu(D,Y)\geqslant1)}\mu(D,Y)\cdot\mathbf{1}_{(X,D)}.\end{equation}
It then follows from Lemma 5.4 that (1) holds true. Now we choose $I\subseteq K$, $J\subseteq L$ such that $|\{i\in I\mid U_{i}=C\}|=-\mu(C,Y)$ for all $C\in\Gamma$ with $\mu(C,Y)\leqslant-1$, and $|\{j\in J\mid V_{j}=D\}|=\mu(D,Y)$ for all $D\in\Gamma$ with $\mu(D,Y)\geqslant1$. We note that $|K|=e\cdot|I|$, $|L|=e\cdot|J|$. By (5.5) and (5.6), $((U_{i}\mid i\in I),(V_{j}\mid j\in J))$ is a non-trivial solution. Hence either $I=K$ or $J=L$ holds true, which implies that $e=1$, and (3) immediately follows. Finally, (2) follows from (1), (3) and (1) of Lemma 5.4, as desired.
\end{proof}

\setlength{\parindent}{2em}
Combining Proposition 5.1 and Lemma 5.4, we immediately derive the minimal length of nontrivial solutions with respect to $\Gamma$. The following theorem is the main result of this section.

\setlength{\parindent}{0em}
\begin{theorem}
{Let $\Lambda=\{C\in\Gamma\mid\text{$C\neq\langle\{x\}\rangle_{\Gamma}$ for all $x\in X$}\}$. Suppose that $\emptyset\not\in\Gamma$, $\Lambda\neq\emptyset$, and let $n\triangleq\min\{\frac{1}{2}(\sum_{(U\in\Gamma,U\subseteq W)}|\mu(U,W)|)\mid W\in\Lambda\}$. Then, there exists a non-trivial solution $((U_1,\dots,U_{n}),(V_1,\dots,V_{n}))\in\Gamma^{n}\times\Gamma^{n}$. Moreover, let $p,q\in\mathbb{N}$ and let $((A_{1},\dots,A_{p}),(B_{1},\dots,B_{q}))\in\Gamma^{p}\times\Gamma^{q}$ be a non-trivial solution. Then, we have $p=q\geqslant n$.
}
\end{theorem}

\setlength{\parindent}{2em}
The assumptions that $\emptyset\not\in\Gamma$, $\Lambda\neq\emptyset$ in Theorem 5.1 are essential. Since if $\emptyset\in\Gamma$, then we have $\mathbf{1}_{(X,\emptyset)}=0$, which obviously induces a nontrivial solution; and if $\Lambda=\emptyset$, then all the solutions are necessarily trivial (\textit{c.f.}, [13, Lemma 2.4.1] and [14, Lemma 2]). Moreover, inspired by Theorem 5.1, we give the following definition.

\setlength{\parindent}{0em}
\begin{definition}
{Suppose that $X$ is a finite left $S$-module that has a non-cyclic $S$-submodule, and $(\Gamma,\subseteq)$ is the $S$-submodule lattice of $X$. Let $\Lambda$ be the set of all the non-cyclic $S$-submodules of $X$. We define $\zeta_{S}(X)\triangleq\min\{\frac{1}{2}(\sum_{(U\in\Gamma,U\subseteq W)}|\mu(U,W)|)\mid W\in\Lambda\}$.
}
\end{definition}

\setlength{\parindent}{2em}
\begin{remark}
{Consider a finite field $\mathbb{F}$ and $e\in\mathbb{Z}^{+}$. Suppose that $S=Mat_{e}(\mathbb{F})$, $X$ is a finite non-cyclic left $S$-module, and $(\Gamma,\subseteq)$ is the $S$-submodule lattice of $X$. Then, via some computation, we have $\zeta_{S}(X)=\prod_{i=1}^{e}(|\mathbb{F}|^{i}+1)$. Hence Theorem 5.1 generalizes Lemma 5.2.
}
\end{remark}

At the end of this section, we use Theorem 5.1 and some known results in \cite{13,29} to establish some sufficient conditions for Hamming weight preserving maps to be extendable, which will also be applied to weighted poset metric in the next section.

\setlength{\parindent}{0em}
\begin{lemma}
{Let $G$ be a group, and let $A,B,C,D$ be subgroups of $G$. Then, the following three statements are equivalent to each other:

{\bf{(1)}}\,\,Either $(A=C,B=D)$ or $(A=D,B=C)$ holds true;

{\bf{(2)}}\,\,$\mathbf{1}_{(G,A)}+\mathbf{1}_{(G,B)}=\mathbf{1}_{(G,C)}+\mathbf{1}_{(G,D)}$;

{\bf{(3)}}\,\,$A\cup B=C\cup D$, $A\cap B=C\cap D$.

}
\end{lemma}

\begin{proof}
We infer that $(1)\Longrightarrow(2)$ and $(2)\Longrightarrow(3)$ are straightforward to verify, and $(3)\Longrightarrow(1)$ follows from the fact that for three subgroups $U,V,W\leqslant G$, if $U\subseteq V\cup W$, then either $U\subseteq V$ or $U\subseteq W$ holds true. We omit the details of the verification.
\end{proof}

\setlength{\parindent}{0em}
\begin{proposition}
{Let $M$ be a strong pseudo-injective left $S$-module, and fix $n\in\mathbb{Z}^{+}$. Let $C$ be an $S$-submodule of $M^{n}$, and assume that one of the following five conditions holds:

{\bf{(1)}}\,\,$C$ is finite and $\soc_{S}(M)$ is a cyclic left $S$-module;

{\bf{(2)}}\,\,All the $S$-submodules of $C$ are cyclic;

{\bf{(3)}}\,\,$C$ is finite, $C$ has a non-cyclic $S$-submodule and $n\leqslant\zeta_{S}(C)-1$;

{\bf{(4)}}\,\,$n\leqslant2$;

{\bf{(5)}}\,\,For any proper ideal $Q$ of $S$, it holds that $S/Q$ is infinite.

Then, for any $f\in\Hom_{S}(C,M^{n})$ such that $f$ preserves Hamming weight, there exists a Hamming weight isometry $\psi\in\Aut_{S}(M^{n})$ such that $\psi\mid_{C}=f$.
}
\end{proposition}

\begin{proof}
For any $i\in[1,n]$, define $\varepsilon_{i}:M^{n}\longrightarrow M$ as $\varepsilon_{i}(\alpha)=\alpha_{i}$. Let $f\in\Hom_{S}(C,M^{n})$ be a Hamming weight preserving map. By Lemma 5.1, $\mathcal{U}\triangleq((\ker(\varepsilon_{i}\circ f)\mid i\in[1,n]),(\ker(\varepsilon_{i})\cap C\mid i\in[1,n]))$ is a solution. If (1) holds, then for any $i\in[1,n]$, the left $S$-modules $C/\ker(\varepsilon_{i})\cap C$ and $C/\ker(\varepsilon_{i}\circ f)$ have cyclic socles, and hence $\mathcal{U}$ is trivial by the proof of [13, Theorem 2.1.3]; if (2) holds, then $\mathcal{U}$ is trivial by [13, Lemma 2.4.1]; if (3) holds, then $\mathcal{U}$ is trivial by Theorem 5.1; if (4) holds, then $\mathcal{U}$ is trivial by Lemma 5.5; and if (5) holds, then $\mathcal{U}$ is trivial by [29, Lemma 4.4]. Hence by Lemma 5.1, $f$ extends to a Hamming weight isometry of $M^{n}$, as desired.
\end{proof}

\section{The MEP for weighted poset metric}
\setlength{\parindent}{2em}
Throughout this section, we let $\mathbf{P}=(\Omega,\preccurlyeq_{\mathbf{P}})$ be a poset, $m$ be the largest cardinality of a chain in $\mathbf{P}$, and for any $r\in[1,m]$, we define $W_{r}\triangleq\{u\in\Omega\mid \len_{\mathbf{P}}(u)=r\}$. We also fix $\omega:\Omega\longrightarrow\mathbb{R}^{+}$, and define $\varpi:2^{\Omega}\longrightarrow\mathbb{R}$ as $\varpi(A)=\sum_{a\in A}\omega(a)$. In addition, for any $\alpha\in\mathbf{H}$ and $J\subseteq\Omega$, we define $\alpha\mid_{J}\in\prod_{i\in J}H_{i}$ as $(\alpha\mid_{J})_{j}=\alpha_{j}$ for all $j\in J$.

\subsection{Some necessary and sufficient conditions}

\setlength{\parindent}{0em}
\begin{lemma}
{Assume that for any linear code $C\subseteq\mathbf{H}$ and $\chi\in\Hom_{S}(C,\mathbf{H})$ such that $\chi$ preserves $\mathbf{P}$-support, there exists $\varphi\in\GL_{(\mathbf{P},\omega)}(\mathbf{H})$ with $\varphi\mid_{C}=\chi$. Then, we have:

{\bf{(1)}}\,\,If either $H_{i}\neq\{0\}$ for all $i\in\Omega$ or $\mathbf{P}$ is hierarchical, then $\mathbf{H}$ satisfies Condition (A);

{\bf{(2)}}\,\,Let $C\subseteq\mathbf{H}$ be a linear code, and fix $f\in\Hom_{S}(C,\mathbf{H})$, $\lambda\in\Aut(\mathbf{P})$. Suppose that $\omega(i)=\omega(\lambda(i))$, $H_{i}\cong H_{\lambda(i)}$ for all $i\in\Omega$, and $\langle\supp(f(\alpha))\rangle_{\mathbf{P}}=\lambda[\langle\supp(\alpha)\rangle_{\mathbf{P}}]$ for all $\alpha\in C$. Then, there exists $\psi\in\GL_{(\mathbf{P},\omega)}(\mathbf{H})$ with $\psi\mid_{C}=f$;

{\bf{(3)}}\,\,Let $\gamma,\theta\in\mathbf{H}$ and $\mu\in\Aut(\mathbf{P})$. Suppose that $\omega(i)=\omega(\mu(i))$, $H_{i}\cong H_{\mu(i)}$ for all $i\in\Omega$, and $\langle\supp(a\cdot\theta)\rangle_{\mathbf{P}}=\mu[\langle\supp(a\cdot\gamma)\rangle_{\mathbf{P}}]$ for all $a\in S$. Then, there exists $\varphi\in\GL_{(\mathbf{P},\omega)}(\mathbf{H})$ with $\varphi(\gamma)=\theta$.
}
\end{lemma}

\begin{proof}
{\bf{(1)}}\,\,With the help of Corollary 3.1, the proof is similar to those of the second part of Lemma 4.1, and hence we omit the details.

{\bf{(2)}}\,\,Choose $(\rho_{i}\mid i\in\Omega)\in\prod_{i\in\Omega}\Aut_{S}(H_{i},H_{\lambda(i)})$, and define $\varphi:\mathbf{H}\longrightarrow\mathbf{H}$ such that for any $\alpha\in \mathbf{H}$, $\varphi(\alpha)_{\lambda(i)}=\rho_{i}(\alpha_{i})$ for all $i\in\Omega$. It is straightforward to verify that $\varphi\in\GL_{(\mathbf{P},\omega)}(\mathbf{H})$ and $\varphi^{-1}\circ f\in\Hom_{S}(C,\mathbf{H})$ preserves $\mathbf{P}$-support. Hence we can choose $\sigma\in\GL_{(\mathbf{P},\omega)}(\mathbf{H})$ such that $\sigma\mid_{C}=\varphi^{-1}\circ f$. It follows that $\varphi\circ\sigma\in\GL_{(\mathbf{P},\omega)}(\mathbf{H})$ and $(\varphi\circ \sigma)\mid_{C}=f$, as desired.

{\bf{(3)}}\,\,Applying (2) to $f\in\Hom_{S}(S\cdot\gamma,\mathbf{H})$ defined as $f(\gamma)=\theta$, the result immediately follows.
\end{proof}

\setlength{\parindent}{2em}
Now we show that if $\mathbf{H}$ satisfies Condition (C), then the MEP for $(\mathbf{P},\omega)$-weight implies Conditions (A), (B) and (D). The following theorem is the first main result of this subsection.

\setlength{\parindent}{0em}
\begin{theorem}
{{\bf{(1)}}\,\,Assume that $\mathbf{H}$ satisfies Condition (C) and for any $\gamma,\theta\in\mathbf{H}$ such that $\wt_{(\mathbf{P},\omega)}(a\cdot\gamma)=\wt_{(\mathbf{P},\omega)}(a\cdot\theta)$ for all $a\in S$, there exists $\psi\in\GL_{(\mathbf{P},\omega)}(\mathbf{H})$ such that $\psi(\gamma)=\theta$. Then, $(\mathbf{H},(\mathbf{P},\omega))$ satisfies Condition (D).

{\bf{(2)}}\,\,Suppose that $\mathbf{H}$ satisfies the MEP for $(\mathbf{P},\omega)$-weight. Then, for $\gamma,\theta\in\mathbf{H}$ such that $\wt_{(\mathbf{P},\omega)}(a\cdot\gamma)=\wt_{(\mathbf{P},\omega)}(a\cdot\theta)$ for all $a\in S$, there exists $\psi\in\GL_{(\mathbf{P},\omega)}(\mathbf{H})$ such that $\psi(\gamma)=\theta$. Further assume that $\mathbf{H}$ satisfies Condition (C). Then, $\mathbf{H}$ satisfies Condition (A), $(\mathbf{H},\mathbf{P})$ satisfies Condition (B), and $(\mathbf{H},(\mathbf{P},\omega))$ satisfies Condition (D).
}
\end{theorem}

\begin{proof}
{\bf{(1)}}\,\,We fix $\xi\in\mathbf{H}$ such that $\xi\neq0$ and for any $k,l\in\Omega$, it holds that
$(\forall~a\in S:a\cdot\xi_{k}=0\Longleftrightarrow a\cdot\xi_{l}=0)$. Consider $I,J\in\mathcal{I}(\mathbf{P})$ with $\sum_{i\in I}\omega(i)=\sum_{j\in J}\omega(j)$. Since $\supp(\xi)=\Omega$, there uniquely exists $\gamma,\theta\in\mathbf{H}$ such that $\supp(\gamma)=I$, $\gamma_{i}=\xi_{i}$ for all $i\in I$, $\supp(\theta)=J$, $\theta_{j}=\xi_{j}$ for all $j\in J$. For an arbitrary $a\in S$, considering $I=\emptyset$ and $I\neq\emptyset$ separately, we deduce that either $\supp(a\cdot\gamma)=I$, $\supp(a\cdot\theta)=J$ or $a\cdot\gamma=a\cdot\theta=0$ holds true, which further implies that $\wt_{(\mathbf{P},\omega)}(a\cdot\gamma)=\wt_{(\mathbf{P},\omega)}(a\cdot\theta)$. Hence we can choose $\psi\in\GL_{(\mathbf{P},\omega)}(\mathbf{H})$ with $\psi(\gamma)=\theta$. By Corollary 3.1, we can further choose $\lambda\in\Aut(\mathbf{P})$ such that $\omega(i)=\omega(\lambda(i))$, $H_{i}\cong H_{\lambda(i)}$ for all $i\in\Omega$, and $\langle\supp(\psi(\alpha))\rangle_{\mathbf{P}}=\lambda[\langle\supp(\alpha)\rangle_{\mathbf{P}}]$ for all $\alpha\in\mathbf{H}$. From $\psi(\gamma)=\theta$, we deduce that $J=\lambda[I]$. The above discussion yields that $(\mathbf{P},\omega)$ satisfies the UDP. Next, we fix $u,v\in\Omega$ such that $\len_{\mathbf{P}}(u)=\len_{\mathbf{P}}(v)\triangleq r$ and $\omega(u)=\omega(v)$. Let $I_1=(\bigcup_{j=1}^{r-1}W_{j})\cup\{u\}$, $J_1=(\bigcup_{j=1}^{r-1}W_{j})\cup\{v\}$. It follows that $I_1,J_1\in\mathcal{I}(\mathbf{P})$ and $\sum_{i\in I_1}\omega(i)=\sum_{j\in J_1}\omega(j)$. Hence we can choose $\mu\in\Aut(\mathbf{P})$ such that $J_1=\mu[I_1]$ and $H_{i}\cong H_{\mu(i)}$ for all $i\in\Omega$. Apparently, we have $v=\mu(u)$, which further implies that $H_{u}\cong H_{\mu(u)}=H_{v}$, as desired.

{\bf{(2)}}\,\,Let $\gamma,\theta\in\mathbf{H}$ such that $\wt_{(\mathbf{P},\omega)}(a\cdot\gamma)=\wt_{(\mathbf{P},\omega)}(a\cdot\theta)$ for all $a\in S$. Then, there uniquely exists $f\in\Hom_{S}(S\cdot\gamma,\mathbf{H})$ defined as $f(\gamma)=\theta$. Furthermore, one can check that $f$ preserves $(\mathbf{P},\omega)$-weight. Hence we can choose $\psi\in\GL_{(\mathbf{P},\omega)}(\mathbf{H})$ such that $\psi\mid_{S\cdot\gamma}=f$. It follows that $\psi(\gamma)=f(\gamma)=\theta$, as desired. Now the rest immediately follows from (1), (1) of Lemma 6.1 and the first part of Lemma 4.1.
\end{proof}

\setlength{\parindent}{2em}
From now on, we will focus on the case that either $\mathbf{P}$ is hierarchical or $\omega$ is identically $1$, i.e., the $\mathbf{P}$-weight case.

\setlength{\parindent}{0em}
\begin{lemma}
{$\prod_{(i\in\Omega,\omega(i)=b)}H_{i}$ satisfies the MEP for Hamming weight for all $b\in\omega[\Omega]$ if and only if for any linear code $C\subseteq\mathbf{H}$ and $f\in\Hom_{S}(C,\mathbf{H})$ such that
\begin{equation}\small\forall~\alpha\in C,b\in\mathbb{R}:|\{i\in\supp(f(\alpha))\mid\omega(i)=b\}|=|\{i\in\supp(\alpha)\mid \omega(i)=b\}|,\end{equation}
there exists $\varphi\in\GL_{((\Omega,=),\omega)}(\mathbf{H})$ such that $\varphi\mid_{C}=f$. Consequently, if $\mathbf{H}$ satisfies the MEP for $((\Omega,=),\omega)$-weight, then $\prod_{(i\in\Omega,\omega(i)=b)}H_{i}$ satisfies the MEP for Hamming weight for all $b\in\omega[\Omega]$. Conversely, if $((\Omega,=),\omega)$ satisfies the UDP, and $\prod_{(i\in\Omega,\omega(i)=b)}H_{i}$ satisfies the MEP for Hamming weight for all $b\in\omega[\Omega]$, then $\mathbf{H}$ satisfies the MEP for $((\Omega,=),\omega)$-weight.
}
\end{lemma}

\begin{proof}
First, we prove the ``if'' part. Consider $c\in\omega[\Omega]$, and write $M\triangleq\prod_{(i\in\Omega,\omega(i)=c)}H_{i}$. Let $B$ be an $S$-submodule of $M$, and let $\rho\in\Hom_{S}(B,M)$ preserve Hamming weight. Define $C\triangleq\{\sum_{(i\in\Omega,\omega(i)=c)}\eta_{i}(\beta_{i})\mid\beta\in B\}$, and define $f\in\Hom_{S}(C,\mathbf{H})$ as $f(\sum_{(i\in\Omega,\omega(i)=c)}\eta_{i}(\beta_{i}))=\sum_{(i\in\Omega,\omega(i)=c)}\eta_{i}(\rho(\beta)_{i})$ for all $\beta\in B$. We note that $f$ satisfies (6.1). Hence we can choose $\varphi\in\GL_{((\Omega,=),\omega)}(\mathbf{H})$ with $\varphi\mid_{C}=f$. By Corollary 3.1, (6.1) holds true for $\mathbf{H}$ and $\varphi$. Hence there uniquely exists a Hamming weight isometry $\varepsilon\in\Aut_{S}(M)$ defined as $\varepsilon(\alpha\mid_{\{i\in\Omega:\omega(i)=c\}})=\varphi(\alpha)\mid_{\{i\in\Omega:\omega(i)=c\}}$ for all $\alpha\in\mathbf{H}$. Moreover, it is straightforward to verify that $\varepsilon\mid_{B}=\rho$, as desired.

\hspace*{4mm}\,\,Second, we prove the ``only if'' part. Let $C\subseteq\mathbf{H}$ be a linear code and let $f\in\Hom_{S}(C,\mathbf{H})$ satisfy (6.1). Consider an arbitrary $c\in\omega[\Omega]$. Define $A_{c}\triangleq\{\alpha\mid_{\{i\in\Omega:\omega(i)=c\}}\mid \alpha\in C\}$. Since $f$ satisfies (6.1), there uniquely exists $\rho_{c}\in\Hom_{S}(A_{c},\prod_{(i\in\Omega,\omega(i)=c)}H_{i})$ defined as $\rho_{c}(\alpha\mid_{\{i\in\Omega:\omega(i)=c\}})=f(\alpha)\mid_{\{i\in\Omega:\omega(i)=c\}}$ for all $\alpha\in C$. Moreover, we note that $\rho_{c}$ preserves Hamming weight. Hence we can choose a Hamming weight isometry $\varepsilon_{c}$ of $\prod_{(i\in\Omega,\omega(i)=c)}H_{i}$ such that $\varepsilon_{c}\mid_{A_{c}}=\rho_{c}$. Now define $\varphi\in\End_{S}(\mathbf{H})$ such that for any $\alpha\in \mathbf{H}$, we have $\varphi(\alpha)\mid_{\{i\in\Omega:\omega(i)=c\}}=\varepsilon_{c}(\alpha\mid_{\{i\in\Omega:\omega(i)=c\}})$ for all $c\in\omega[\Omega]$. It is straightforward to verify that $\varphi\in\GL_{((\Omega,=),\omega)}(\mathbf{H})$ and $\varphi\mid_{C}=f$, as desired.

\hspace*{4mm}\,\,Finally, we note that for a linear code $C\subseteq\mathbf{H}$ and $f\in\Hom_{S}(C,\mathbf{H})$, if $f$ satisfies (6.1), then $f$ preserves $((\Omega,=),\omega)$-weight; and if $((\Omega,=),\omega)$ satisfies the UDP and $f$ preserves $((\Omega,=),\omega)$-weight, then $f$ satisfies (6.1). Hence the rest immediately follows from the two proven parts.
\end{proof}

\setlength{\parindent}{2em}
Now we are ready to give the connections between the MEP for $(\mathbf{P},\omega)$-weight and the MEP for Hamming weight. The following theorem is the second main result of this subsection.

\setlength{\parindent}{0em}
\begin{theorem}
{Assume that $\mathbf{P}$ is hierarchical. Then, it holds that:

{\bf{(1)}}\,\,$\mathbf{H}$ satisfies the MEP for $(\mathbf{P},\omega)$-weight if and only if $\mathbf{H}$ satisfies Condition (A), $(\mathbf{H},\mathbf{P})$ satisfies Condition (B), and for any $r\in[1,m]$, $\prod_{i\in W_{r}}H_{i}$ satisfies the MEP for $((W_{r},=),\omega\mid_{W_{r}})$-weight;

{\bf{(2)}}\,\,If $\mathbf{H}$ satisfies the MEP for $(\mathbf{P},\omega)$-weight, then for any $r\in[1,m]$ and $b\in\omega[W_{r}]$, $\prod_{(i\in W_{r},\omega(i)=b)}H_{i}$ satisfies the MEP for Hamming weight;

{\bf{(3)}}\,\,If $(\mathbf{P},\omega)$ satisfies the UDP, $\mathbf{H}$ satisfies Condition (A), $(\mathbf{H},\mathbf{P})$ satisfies Condition (B), and for any $r\in[1,m]$, $b\in\omega[W_{r}]$, $\prod_{(i\in W_{r},\omega(i)=b)}H_{i}$ satisfies the MEP for Hamming weight, then $\mathbf{H}$ satisfies the MEP for $(\mathbf{P},\omega)$-weight;

{\bf{(4)}}\,\,Suppose that $\mathbf{H}$ satisfies Condition (C). Then, $\mathbf{H}$ satisfies the MEP for $(\mathbf{P},\omega)$-weight if and only if $\mathbf{H}$ satisfies Condition (A), $(\mathbf{H},\mathbf{P})$ satisfies Condition (B), $(\mathbf{H},(\mathbf{P},\omega))$ satisfies Condition (D), and for any $r\in[1,m]$, $b\in\omega[W_{r}]$, $\prod_{(i\in W_{r},\omega(i)=b)}H_{i}$ satisfies the MEP for Hamming weight.
}
\end{theorem}

\begin{proof}
{\bf{(1)}}\,\,First, we prove the ``only if'' part. The first two assertions follow from Lemmas 6.1 and 4.1, respectively. Now consider $r\in[1,m]$. Let $D$ be an $S$-submodule of $\prod_{i\in W_{r}}H_{i}$, and let $g\in\Hom_{S}(D,\prod_{i\in W_{r}}H_{i})$ such that $g$ preserves $((W_{r},=),\omega\mid_{W_{r}})$-weight. Define $C\triangleq\{\sum_{i\in W_{r}}\eta_{i}(\gamma_{i})\mid \gamma\in D\}$, and define $f\in\Hom_{S}(C,\mathbf{H})$ as $f(\sum_{i\in W_{r}}\eta_{i}(\gamma_{i}))=\sum_{i\in W_{r}}\eta_{i}(g(\gamma)_{i})$ for all $\gamma\in D$. Since $\mathbf{P}$ is hierarchical, $f$ preserves $(\mathbf{P},\omega)$-weight. Hence we can choose $\varphi\in\GL_{(\mathbf{P},\omega)}(\mathbf{H})$ such that $\varphi\mid_{C}=f$. Define $\tau\in\End_{S}(\prod_{i\in W_{r}}H_{i})$ as $\tau(\gamma)=\varphi(\sum_{i\in W_{r}}\eta_{i}(\gamma_{i}))\mid_{W_{r}}$. It follows from the facts $\varphi\in\GL_{(\mathbf{P},\omega)}(\mathbf{H})$ and $\mathbf{P}$ is hierarchical that $\tau$ is a $((W_{r},=),\omega\mid_{W_{r}})$-weight isometry of $\prod_{i\in W_{r}}H_{i}$. Finally, from $\varphi\mid_{C}=f$, one can check that $\tau\mid_{D}=g$, as desired.

\hspace*{4mm}\,\,Second, we prove the ``if'' part. Let $C\subseteq\mathbf{H}$ be a linear code, and let $f\in\Hom_{S}(C,\mathbf{H})$ such that $f$ preserves $(\mathbf{P},\omega)$-weight. Fix $r\in[1,m]$ such that $C\subseteq\delta(\bigcup_{j=1}^{r}W_{j})$, and let $C_1\triangleq C\cap\delta(\bigcup_{j=1}^{r-1}W_{j})$. Since $\mathbf{P}$ is hierarchical, the following two statements hold true:

$(i)$\,\,$f[C]\subseteq\delta(\bigcup_{j=1}^{r}W_{j})$, $f^{-1}[\delta(\bigcup_{j=1}^{r-1}W_{j})]=C_1$;

$(ii)$\,\,$\varpi(\supp(\alpha)\cap W_{r})=\varpi(\supp(f(\alpha))\cap W_{r})$ for all $\alpha\in C$.

By induction, we assume that $f\mid_{C_1}$ extends to a $(\mathbf{P},\omega)$-weight isometry of $\mathbf{H}$. By Corollary 3.1, we can choose $\mu\in\Aut(\mathbf{P})$ such that $\omega(i)=\omega(\mu(i))$, $H_{i}\cong H_{\mu(i)}$ for all $i\in\Omega$ and $\langle\supp(f(\alpha))\rangle_{\mathbf{P}}=\mu[\langle\supp(\alpha)\rangle_{\mathbf{P}}]$ for all $\alpha\in C_{1}$. Now let $D=\{\alpha\mid_{W_{r}}\mid \alpha\in C\}$. By $(i)$ and $(ii)$, there uniquely exists $g\in\Hom_{S}(D,\prod_{i\in W_{r}}H_{i})$ defined as $g(\alpha\mid_{W_{r}})=f(\alpha)\mid_{W_{r}}$ for all $\alpha\in C$. Moreover, $g$ preserves $((W_{r},=),\omega\mid_{W_{r}})$-weight. Hence we can choose a $((W_{r},=),\omega\mid_{W_{r}})$-weight isometry $\rho\in\Aut_{S}(\prod_{i\in W_{r}}H_{i})$ such that $\rho\mid_{D}=g$. Applying Corollary 3.1 to $\prod_{i\in W_{r}}H_{i}$, $((W_{r},=),\omega\mid_{W_{r}})$ and $\rho$, we can choose a permutation $\sigma$ of $W_{r}$ such that $\omega(i)=\omega(\sigma(i))$, $H_{i}\cong H_{\sigma(i)}$ for all $i\in W_{r}$ and $\supp(g(\gamma))=\sigma[\supp(\gamma)]$ for all $\gamma\in D$. From the definition of $D$ and $g$, we infer that $\supp(f(\alpha))\cap W_{r}=\sigma[\supp(\alpha)\cap W_{r}]$ for all $\alpha\in C$. Now define $\lambda:\Omega\longrightarrow\Omega$ as $\lambda\mid_{W_{r}}=\sigma$ and $\lambda\mid_{\Omega-W_{r}}=\mu\mid_{\Omega-W_{r}}$. Since $\mathbf{P}$ is hierarchical, we have $\lambda\in\Aut(\mathbf{P})$. Moreover, it can be readily verified that $\omega(i)=\omega(\lambda(i))$, $H_{i}\cong H_{\lambda(i)}$ for all $i\in \Omega$. Now for an arbitrary $\alpha\in C$, we will show that $\langle\supp(f(\alpha))\rangle_{\mathbf{P}}=\lambda[\langle\supp(\alpha)\rangle_{\mathbf{P}}]$. If $\alpha\in C_1$, then we have $\langle\supp(f(\alpha))\rangle_{\mathbf{P}}=\mu[\langle\supp(\alpha)\rangle_{\mathbf{P}}]$ and $\langle\supp(\alpha)\rangle_{\mathbf{P}}\subseteq\bigcup_{j=1}^{r-1}W_{j}$, which, along with $\lambda\mid_{\Omega-W_{r}}=\mu\mid_{\Omega-W_{r}}$, implies that $\langle\supp(f(\alpha))\rangle_{\mathbf{P}}=\lambda[\langle\supp(\alpha)\rangle_{\mathbf{P}}]$, as desired. Therefore in the following, we assume that $\alpha\in C-C_1$. By $\lambda\mid_{W_{r}}=\sigma$, we have $\supp(f(\alpha))\cap W_{r}=\lambda[\supp(\alpha)\cap W_{r}]$. Moreover, by $(i)$, we have $\supp(\alpha)\subseteq\bigcup_{j=1}^{r}W_{j}$, $\supp(\alpha)\nsubseteq\bigcup_{j=1}^{r-1}W_{j}$, $\supp(f(\alpha))\subseteq\bigcup_{j=1}^{r}W_{j}$, $\supp(f(\alpha))\nsubseteq\bigcup_{j=1}^{r-1}W_{j}$. Since $\mathbf{P}$ is hierarchical and $\lambda\in\Aut(\mathbf{P})$, we have
\begin{eqnarray*}
\begin{split}
\langle\supp(f(\alpha))\rangle_{\mathbf{P}}&=\langle\supp(f(\alpha))\cap W_{r}\rangle_{\mathbf{P}}=\langle\lambda[\supp(\alpha)\cap W_{r}]\rangle_{\mathbf{P}}\\
&=\lambda[\langle\supp(\alpha)\cap W_{r}\rangle_{\mathbf{P}}]=\lambda[\langle\supp(\alpha)\rangle_{\mathbf{P}}],
\end{split}
\end{eqnarray*}
as desired. Now Theorem 4.1 implies that $\mathbf{H}$ satisfies the MEP for $\mathbf{P}$-support, which, along with (2) of Lemma 6.1, further implies that $f$ extends to a $(\mathbf{P},\omega)$-weight isometry of $\mathbf{H}$, as desired.

{\bf{(2)}}\,\,For $r\in[1,m]$, it follows from (1) that $\prod_{i\in W_{r}}H_{i}$ satisfies the MEP for $((W_{r},=),\omega\mid_{W_{r}})$-weight, and hence the desired result follows from applying Lemma 6.2 to $\prod_{i\in W_{r}}H_{i}$ and $((W_{r},=),\omega\mid_{W_{r}})$.

{\bf{(3)}}\,\,Consider an arbitrary $r\in[1,m]$. We first show that $((W_{r},=),\omega\mid_{W_{r}})$ satisfies the UDP. Let $U,V\subseteq W_{r}$ with $\varpi(U)=\varpi(V)$. Since $\mathbf{P}$ is hierarchical, we have $\varpi(\langle U\rangle_{\mathbf{P}})=\varpi(\langle V\rangle_{\mathbf{P}})$. Hence we can choose $\lambda\in\Aut(\mathbf{P})$ such that $\langle V\rangle_{\mathbf{P}}=\lambda[\langle U\rangle_{\mathbf{P}}]$ and $\omega(i)=\omega(\lambda(i))$ for all $i\in\Omega$. It follows that $\mu\triangleq\lambda\mid_{W_{r}}$ is a permutation of $W_{r}$, $V=\mu[U]$ and $\omega(i)=\omega(\mu(i))$ for all $i\in W_{r}$, as desired. Applying Lemma 6.2 to $\prod_{i\in W_{r}}H_{i}$ and $((W_{r},=),\omega\mid_{W_{r}})$, we deduce that $\prod_{i\in W_{r}}H_{i}$ satisfies the MEP for $((W_{r},=),\omega\mid_{W_{r}})$-weight. Now the desired result immediately follows from (1).

{\bf{(4)}}\,\,The ``if'' part follows from (3), and the ``only if'' part follows from (2) along with (2) of Theorem 6.1, as desired.
\end{proof}

\setlength{\parindent}{2em}
In the following corollary, we apply Theorems 6.1 and 6.2 to $\mathbf{P}$-weight. In particular, we show that if $\mathbf{H}$ satisfies Condition (C), then the MEP for $\mathbf{P}$-weight can be discussed without assuming $\mathbf{P}$ to be hierarchical.

\setlength{\parindent}{0em}
\begin{corollary}
{{\bf{(1)}}\,\,Suppose that $\mathbf{P}$ is hierarchical. Then, $\mathbf{H}$ satisfies the MEP for $\mathbf{P}$-weight if and only if $\mathbf{H}$ satisfies Condition (A), $(\mathbf{H},\mathbf{P})$ satisfies Condition (B), and for any $r\in[1,m]$, $\prod_{i\in W_{r}}H_{i}$ satisfies the MEP for Hamming weight.

{\bf{(2)}}\,\,Suppose that $\mathbf{H}$ satisfies Condition (C). Then, $\mathbf{H}$ satisfies the MEP for $\mathbf{P}$-weight if and only if $\mathbf{H}$ satisfies Condition (A), $(\mathbf{H},\mathbf{P})$ satisfies Conditions (B) and (E), and for any $r\in[1,m]$, $\prod_{i\in W_{r}}H_{i}$ satisfies the MEP for Hamming weight.
}
\end{corollary}

\begin{proof}
Part (1) follows from applying (1) of Theorem 6.2 to the constant $1$ map. Moreover, the ``if'' part of (2) follows from (1), and the ``only if'' part of (2) follows from applying (2) of Theorem 6.1 to the constant $1$ map, along with (1) and Lemma 2.1.
\end{proof}

\setlength{\parindent}{2em}
Next, combining Theorem 6.2 and Proposition 5.2, we give some sufficient conditions for the MEP for $(\mathbf{P},\omega)$-weight.

\setlength{\parindent}{0em}
\begin{theorem}
{Suppose that $\mathbf{P}$ is hierarchical, $\mathbf{H}$ satisfies Condition (A), $(\mathbf{H},\mathbf{P})$ satisfies Condition (B), and $(\mathbf{H},(\mathbf{P},\omega))$ satisfies Condition (D). Further assume that for any $r\in[1,m]$ and $b\in\omega[W_{r}]$, one of the following five conditions holds:

{\bf{(1)}}\,\,For any $i\in W_{r}$ such that $\omega(i)=b$, $H_{i}$ is finite and $\soc_{S}(H_{i})$ is cyclic;

{\bf{(2)}}\,\,All the $S$-submodules of $\prod_{(i\in W_{r},\omega(i)=b)}H_{i}$ are cyclic;

{\bf{(3)}}\,\,$\prod_{(i\in W_{r},\omega(i)=b)}H_{i}$ is finite and has a non-cyclic $S$-submodule, and it holds that $|\{i\in W_{r}\mid\omega(i)=b\}|\leqslant\zeta_{S}(\prod_{(i\in W_{r},\omega(i)=b)}H_{i})-1$;

{\bf{(4)}}\,\,$|\{i\in W_{r}\mid\omega(i)=b\}|\leqslant2$;

{\bf{(5)}}\,\,For any proper ideal $Q$ of $S$, it holds that $S/Q$ is infinite.

Then, $\mathbf{H}$ satisfies the MEP for $(\mathbf{P},\omega)$-weight.
}
\end{theorem}

\begin{proof}
Fixing $r\in[1,m]$, $b\in\omega[W_{r}]$, and let $n\triangleq|\{i\in W_{r}\mid \omega(i)=b\}|$. Since $\mathbf{H}$ satisfies Condition (A) and $(\mathbf{H},(\mathbf{P},\omega))$ satisfies Condition (D), we can choose a strong pseudo-injective left $S$-module $M$ such that $H_{i}\cong M$ for all $i\in W_{r}$ with $\omega(i)=b$. Moreover, we identify $\prod_{(i\in W_{r},\omega(i)=b)}H_{i}$ with $M^{n}$. For an arbitrary $S$-submodule $C\subseteq M^{n}$, we note that for $a\in\{1,2,4,5\}$, if $(a)$ is satisfied, then $(a)$ of Proposition 5.2 holds; and if (3) is satisfied, then either (2) or (3) of Proposition 5.2 holds. By Proposition 5.2, any Hamming weight preserving map $f\in\Hom_{S}(C,M^{n})$ extends to a Hamming weight isometry of $M^{n}$. It follows that $\prod_{(i\in W_{r},\omega(i)=b)}H_{i}$ satisfies the MEP for Hamming weight. Now (3) of Theorem 6.2 concludes the proof.
\end{proof}

\setlength{\parindent}{2em}
Finally, we consider the special case that $S$ is an Artinian simple ring.

\setlength{\parindent}{0em}
\begin{theorem}
{Suppose that $S$ is an Artinian simple ring. Let $e\in\mathbb{Z}^{+}$ and $\mathbb{D}$ be a division ring such that $S$ is isomorphic to $Mat_{e}(\mathbb{D})$. For any finitely generated left $S$-module $X$, let $\len_{S}(X)$ denote the length of its composition series. Further assume that $H_{i}\neq\{0\}$ for all $i\in\Omega$. Then, we have:

{\bf{(1)}}\,\,Assume that $\mathbf{P}$ is hierarchical. Then, $\mathbf{H}$ satisfies the MEP for $(\mathbf{P},\omega)$-weight if and only if $\mathbf{H}$ is finitely generated, $(\mathbf{H},(\mathbf{P},\omega))$ satisfies Condition (D), and when $S$ is finite, for any $r\in[1,m]$, $b\in \omega[W_{r}]$, it holds that either $(\forall~i\in W_{r}~s.t.~\omega(i)=b:\len_{S}(H_{i})\leqslant e)$ or $|\{i\in W_{r}\mid\omega(i)=b\}|\leqslant(\prod_{i=1}^{e}(|\mathbb{D}|^{i}+1))-1$.

{\bf{(2)}}\,\,$\mathbf{H}$ satisfies the MEP for $\mathbf{P}$-weight if and only if $\mathbf{H}$ is finitely generated, $(\mathbf{H},\mathbf{P})$ satisfies Condition (E), and when $S$ is finite, for any $r\in[1,m]$, it holds that either $(\forall~i\in W_{r}:\len_{S}(H_{i})\leqslant e)$ or $|W_{r}|\leqslant(\prod_{i=1}^{e}(|\mathbb{D}|^{i}+1))-1$.
}
\end{theorem}

\begin{proof}
By Remark 2.2, $\mathbf{H}$ satisfies Condition (C). Since $S$ is Artinian simple, $(\mathbf{H},\mathbf{P})$ satisfies Condition (B); and moreover, a left $S$-module is strong pseudo-injective if and only if it is finitely generated. Hence $\mathbf{H}$ satisfies Condition (A) if and only if $\mathbf{H}$ is finitely generated. We also note that if $S$ is finite, then $\mathbb{D}$ is necessarily a finite field by Wedderburn's theorem.

{\bf{(1)}}\,\,By (4) of Theorem 6.2, we assume that $\mathbf{H}$ is finitely generated and $(\mathbf{H},(\mathbf{P},\omega))$ satisfies Condition (D). Hence for given $r\in[1,m]$, $b\in\omega[W_{r}]$, there uniquely exists $k\in\mathbb{Z}^{+}$ such that for any $i\in W_{r}$ with $\omega(i)=b$, we have $\len_{S}(H_{i})=k$, and hence $H_{i}\cong Mat_{e,k}(\mathbb{D})$. Now if $S$ is infinite, then the result follows from (5) of Theorem 6.3; and if $S$ is finite, then the result follows from (4) of Theorem 6.2 and Lemma 5.3, as desired.

{\bf{(2)}}\,\,By (2) of Corollary 6.1 and Lemma 2.1, the desired result immediately follows from applying (1) to the constant $1$ map.
\end{proof}

\subsection{Connections with other coding-theoretic properties}

\setlength{\parindent}{2em}
In this subsection, we assume that $S=\mathbb{F}$ is a finite field with $|\mathbb{F}|=q$, $(k_{i}\mid i\in\Omega)$ is a family of positive integers, and $\mathbf{H}=\prod_{i\in\Omega}\mathbb{F}^{k_{i}}$. With respect to $(\mathbf{P},\omega)$-weight, we will compare the MEP with some other coding-theoretic properties including the MacWilliams identity, Fourier-reflexivity of partitions, the UDP and that whether $\GL_{(\mathbf{P},\omega)}(\mathbf{H})$ acts transitively on codewords with the same $(\mathbf{P},\omega)$-weight.

As usual, the inner product $\langle~,~\rangle:\mathbf{H}\times\mathbf{H}\longrightarrow \mathbb{F}$ is defined as $\langle\alpha,\beta\rangle=\sum_{i\in\Omega}\sum_{t=1}^{k_{i}}\alpha_{i,t}\beta_{i,t}$, where for $\alpha\in\mathbf{H}$ and $i\in\Omega$, $\alpha_{i,t}$ denotes the $t$-th entry of $\alpha_{i}\in\mathbb{F}^{k_{i}}$. For any linear code $C\subseteq \mathbf{H}$, we let $C^{\bot}\triangleq\{\beta\in\mathbf{H}\mid\text{$\langle\alpha,\beta\rangle=0$ for all $\alpha\in C$}\}$ denote the dual code of $C$. A \textit{partition} of $\mathbf{H}$ is a collection of nonempty disjoint subsets of $\mathbf{H}$ whose union is $\mathbf{H}$. Consider a partition $\Gamma$ of $\mathbf{H}$. For any $\beta,\theta\in\mathbf{H}$, we write $\beta\sim_{\Gamma}\theta$ if $\beta$ and $\theta$ belong to the same member of $\Gamma$, and for any $D\subseteq\mathbf{H}$, we refer to the sequence $(|D\cap B|\mid B\in\Gamma)$ as the \textit{$\Gamma$-distribution} of $D$.

\setlength{\parindent}{0em}
\begin{definition}
{\bf{(1)}}\,\,We let $\mathcal{Q}(\mathbf{H},\mathbf{P},\omega)$ and $\mathcal{Q}(\mathbf{H},\mathbf{P})$ denote the partitions of $\mathbf{H}$ such that for any $\beta,\theta\in\mathbf{H}$, $\beta\sim_{\mathcal{Q}(\mathbf{H},\mathbf{P},\omega)}\theta$ if and only if $\wt_{(\mathbf{P},\omega)}(\beta)=\wt_{(\mathbf{P},\omega)}(\theta)$, and $\beta\sim_{\mathcal{Q}(\mathbf{H},\mathbf{P})}\theta$  if and only if $\wt_{\mathbf{P}}(\beta)=\wt_{\mathbf{P}}(\theta)$. Moreover, the partitions $\mathcal{Q}(\mathbf{H},\mathbf{\overline{P}},\omega)$ and $\mathcal{Q}(\mathbf{H},\mathbf{\overline{P}})$ are defined in a parallel fashion.

{\bf{(2)}}\,\,We say that $(\mathbf{P},\omega)$ admits MacWilliams identity if for any two linear codes $C_1,C_2$ with the same $\mathcal{Q}(\mathbf{H},\mathbf{\overline{P}},\omega)$-distribution, ${C_1}^{\bot}$ and ${C_2}^{\bot}$ have the same $\mathcal{Q}(\mathbf{H},\mathbf{P},\omega)$-distribution. We say that $\mathbf{P}$ admits MacWilliams identity if $(\mathbf{P},\omega)$ admits MacWilliams identity when $\omega$ is identically $1$.

{\bf{(3)}}\,\,Let $\chi$ be a nontrivial additive character of $\mathbb{F}$. For a partition $\Gamma$ of $\mathbf{H}$, let $\textbf{\textit{l}}(\Gamma)$ denote the partition of $\mathbf{H}$ such that for any $\alpha,\gamma\in \mathbf{H}$, $\alpha\sim_{\textbf{\textit{l}}(\Gamma)}\gamma$ if and only if $\sum_{\beta\in B}\chi(\langle\alpha,\beta\rangle)=\sum_{\beta\in B}\chi(\langle\gamma,\beta\rangle)$ for all $B\in\Gamma$. A partition $\Gamma$ of $\mathbf{H}$ is said to be \textit{Fourier-reflexive} if $\textbf{\textit{l}}(\textbf{\textit{l}}(\Gamma))=\Gamma$.
\end{definition}

\setlength{\parindent}{2em}
We note that (2) of Definition 6.1 follows [35, Definition 12], [30, Definition I.2] and [41, Definition 2], and (3) of Definition 6.1 follows [22, Definition 1.2]. It is known that the Fourier-reflexivity of a partition is independent of the choice of the nontrivial additive character (see [20, Theorem 2.4], [22, Page 4]). Hence from now on, we fix a nontrivial additive character $\chi$ of $\mathbb{F}$.

For $\mathbf{P}$-weight, it has been proven in [34, Theorem 3] that if $k_{i}=1$ for all $i\in\Omega$, then each of the aforementioned property is equivalent to $\mathbf{P}$ being hierarchical. For $(\mathbf{P},\omega)$-weight with $\mathbf{P}$ hierarchical, Machado and Firer have given a necessary and sufficient condition for  $(\mathbf{P},\omega)$ to admit MacWilliams identity in [35, Theorem 7], and when $|\mathbb{F}|=2$, they have given a necessary and sufficient condition for the MEP in [35, Theorem 8] (also see [18, Theorems 10 and 12]). Their results show that for binary field alphabet, the MEP is strictly stronger than the MacWilliams identity.

Now we further generalize the above mentioned results. The following theorem and its corollary are the main results of this subsection. We note that a large part of them is known, as detailed in Remark 6.1. Despite such a fact, we gather several properties together and compare them with each other, and show that the MEP is strictly stronger than all the others.

\setlength{\parindent}{0em}
\begin{theorem}
{Consider the following seven statements:

{\bf{(1)}}\,\,$\mathbf{H}$ satisfies the MEP for $(\mathbf{P},\omega)$-weight;

{\bf{(2)}}\,\,For any $\gamma,\theta\in\mathbf{H}$ with $\wt_{(\mathbf{P},\omega)}(\gamma)=\wt_{(\mathbf{P},\omega)}(\theta)$, there exists $\psi\in\GL_{(\mathbf{P},\omega)}(\mathbf{H})$ such that $\psi(\gamma)=\theta$;

{\bf{(3)}}\,\,$(\mathbf{H},(\mathbf{P},\omega))$ satisfies Condition (D);

{\bf{(4)}}\,\,$\mathcal{Q}(\mathbf{H},\mathbf{\overline{P}},\omega)=\textbf{\textit{l}}(\mathcal{Q}(\mathbf{H},\mathbf{P},\omega))$;

{\bf{(5)}}\,\,$(\mathbf{P},\omega)$ admits MacWilliams identity;

{\bf{(6)}}\,\,$\mathcal{Q}(\mathbf{H},\mathbf{P},\omega)$ is Fourier-reflexive;

{\bf{(7)}}\,\,For any $r\in[1,m]$ and $b\in\omega[W_{r}]$, either $(\forall~i\in W_{r}~s.t.~\omega(i)=b:k_{i}=1)$ or $|\{i\in W_{r}\mid\omega(i)=b\}|\leqslant q$ holds true.

Then, we have $(1)\Longrightarrow(2)$, $(2)\Longleftrightarrow(3)$, $(3)\Longrightarrow(4)$, $(4)\Longleftrightarrow(5)$ and $(5)\Longrightarrow(6)$. If $\mathbf{P}$ is hierarchical, then $(1)\Longleftrightarrow((3)\wedge(7))$. If $\mathbf{P}$ is hierarchical and $\omega$ is integer-valued, then $(2)$--$(6)$ are equivalent to each other.
}
\end{theorem}

\begin{proof}
We note that $(1)\Longrightarrow(2)\Longrightarrow(3)$ follows from Theorem 6.1, $(3)\Longrightarrow(4)$ follows from [51, Theorem 10], $(4)\Longleftrightarrow(5)$ follows from [51, Proposition 16], and $(4)\Longrightarrow(6)$ follows from the fact that $|\mathcal{Q}(\mathbf{H},\mathbf{\overline{P}},\omega)|=|\mathcal{Q}(\mathbf{H},\mathbf{P},\omega)|$, together with [20, Theorem 2.4]. Now we prove $(3)\Longrightarrow(2)$. By Corollary 4.1, $\mathbf{H}$ satisfies the MEP for $\mathbf{P}$-support. Let $\gamma,\theta\in\mathbf{H}$ with $\wt_{(\mathbf{P},\omega)}(\gamma)=\wt_{(\mathbf{P},\omega)}(\theta)$. Since $(\mathbf{P},\omega)$ satisfies the UDP, we can choose $\mu\in\Aut(\mathbf{P})$ such that $\langle\supp(\theta)\rangle_{\mathbf{P}}=\mu[\langle\supp(\gamma)\rangle_{\mathbf{P}}]$ and $\omega(i)=\omega(\mu(i))$ for all $i\in\Omega$. For any $i\in\Omega$, it follows from $\omega(i)=\omega(\mu(i))$, $\len_{\mathbf{P}}(i)=\len_{\mathbf{P}}(\mu(i))$ that $k_{i}=k_{\mu(i)}$. By (3) of Lemma 6.1, we can choose $\psi\in\GL_{(\mathbf{P},\omega)}(\mathbf{H})$ such that $\psi(\gamma)=\theta$, which further establishes (2), as desired. Moreover, if  $\mathbf{P}$ is hierarchical, then $(1)\Longleftrightarrow((3)\wedge(7))$ follows from (1) of Theorem 6.4; and if $\mathbf{P}$ is hierarchical and $\omega$ is integer-valued, then by [51, Proposition 16], we have $(6)\Longleftrightarrow(3)$, which further establishes the equivalence between $(2)$--$(6)$, as desired.
\end{proof}

\setlength{\parindent}{0em}
\begin{corollary}
{Consider the following seven statements:

{\bf{(1)}}\,\,$\mathbf{H}$ satisfies the MEP for $\mathbf{P}$-weight;

{\bf{(2)}}\,\,For any $\gamma,\theta\in\mathbf{H}$ with $\wt_{\mathbf{P}}(\gamma)=\wt_{\mathbf{P}}(\theta)$, there exists $\psi\in\Aut_{\mathbb{F}}(\mathbf{H})$ such that $\psi$ preserves $\mathbf{P}$-weight and $\psi(\gamma)=\theta$;

{\bf{(3)}}\,\,$(\mathbf{H},\mathbf{P})$ satisfies Condition (E);

{\bf{(4)}}\,\,$\mathcal{Q}(\mathbf{H},\mathbf{\overline{P}})=\textbf{\textit{l}}(\mathcal{Q}(\mathbf{H},\mathbf{P}))$;

{\bf{(5)}}\,\,$\mathbf{P}$ admits MacWilliams identity;

{\bf{(6)}}\,\,$\mathcal{Q}(\mathbf{H},\mathbf{P})$ is Fourier-reflexive;

{\bf{(7)}}\,\,For any $r\in[1,m]$, either $(\forall~i\in W_{r}:k_{i}=1)$ or $|W_{r}|\leqslant q$ holds true.

Then, we have $(1)\Longleftrightarrow((3)\wedge(7))$ and $(2)\Longleftrightarrow(3)\Longleftrightarrow(4)\Longleftrightarrow(5)\Longleftrightarrow(6)$.
}
\end{corollary}

\begin{proof}
By [50, Theorem II.4], we have $(6)\Longleftrightarrow(3)$. Hence the equivalence between $(2)$--$(6)$ follows from applying Theorem 6.5 to the constant $1$ map. Moreover, $(1)\Longleftrightarrow((3)\wedge(7))$ follows from (2) of Theorem 6.4, as desired.
\end{proof}

\setlength{\parindent}{2em}
\begin{remark}
{In Theorem 6.5, if $\mathbf{P}$ is hierarchical and $\omega$ is integer-valued, then $(2)\Longleftrightarrow(3)$ has been established in [18, Proposition 4] for the case that $\omega(i)=k_{i}$ for all $i\in\Omega$, $(3)\Longleftrightarrow(5)$ has been established in [35, Theorem 7], and $(1)\Longleftrightarrow((3)\wedge(7))$ has been established in [35, Theorem 8] when $q=2$. In Corollary 6.2, $(3)\Longleftrightarrow(4)$ is a special case of [20, Theorems 5.4 and 5.5], and $(3)\Longleftrightarrow(5)$ has been established in [41, Theorems 1 and 2]. Since a partition of $\mathbf{H}$ is Fourier-reflexive if and only if it induces an association scheme (see [52, Theorem 1], [20, Section 2]), if we set $k_{i}=1$ for all $i\in\Omega$, then Corollary 6.2 recovers the equivalence between $\mathbf{P}$ to be hierarchical and parts 1, 2, 3, 4, 6 of [34, Theorem 3]. Based on Theorem 6.5 and Corollary 6.2, we conclude that for weighted poset metric, the MEP is strictly stronger than all the other properties considered in this subsection.
}
\end{remark}

\end{document}